\documentclass [ 10 pt ] {article}
\usepackage{tikz, pgf, pgfplots}
\usetikzlibrary
	{
		3d,
		arrows,									
		arrows.meta,
		automata,								
		backgrounds, 
		bending, 
		calc, 
		chains, 
		datavisualization.formats.functions,		
		decorations.pathmorphing,
		decorations.pathreplacing,
		decorations.text,
		fadings,
		fit, 
		graphs,
		graphs.standard,
		matrix,
		positioning,								
		quantikz2,
		quotes,
		shadings,
		shadows.blur,
		shapes, 
		trees
	}
\pgfplotsset{compat = newest}
\usepgflibrary
	{
		shapes.misc
	}
\usepgfplotslibrary
	{
		dateplot
	}
\usepackage{tikzpeople}
\usepackage{fontawesome5}

\usepackage{amsthm}
\usepackage{amssymb}
\usepackage[bb = boondox]{mathalfa}
\usepackage{dsfont}
\usepackage{newtxmath}
\usepackage{mathtools}
\usepackage{multicol}
\usepackage{braket}
\usepackage{physics}

\usepackage[ compat = newest ]{yquant}

\usepackage{xcolor}
\usepackage{varwidth}
\usepackage[ skins, breakable, listings, theorems ] {tcolorbox}
\usepackage{graphicx}
\usepackage{xurl}                                                                                  %
\usepackage{hyperref}
\hypersetup
	{
		colorlinks,
		linkcolor = {WordRed!75!black},
		citecolor = {WordBlueDarker25},
		urlcolor = {WordAquaDarker50}
	}
\usepackage{colortbl}
\usepackage{float}
\usepackage{nicematrix}
\usepackage{cellspace}
\usepackage{makecell}
\usepackage{multirow}
\usepackage[ linesnumbered, tworuled, vlined ] {algorithm2e}
\usepackage{appendix}
\usepackage{enumitem}
\usepackage{xintexpr}
\usepackage{spreadtab}
\usepackage{framed}
\usepackage{lipsum}
\usepackage{orcidlink}                                                                           %
\newcommand{\orcidicon}[1]{\href{https://orcid.org/#1}{\includegraphics[height=\fontcharht\font`\B]{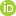}}}

\newtheorem{definition}{Definition}[section]
\newtheorem{theorem}{Theorem}[section]

\newtheorem{proposition}[theorem]{Proposition}
\newtheorem{corollary}[theorem]{Corollary}

\newcounter{mathseed}
\pgfmathsetseed{\arabic{mathseed}}
\pgfdeclarelayer{background}
\pgfsetlayers{background, main}
\pgfdeclaredecoration{irregular fractal line}{init}
{
	\state{init}[width=\pgfdecoratedinputsegmentremainingdistance]
	{
		\pgfpathlineto{\pgfpoint{random*\pgfdecoratedinputsegmentremainingdistance}{(random*\pgfdecorationsegmentamplitude-0.02)*\pgfdecoratedinputsegmentremainingdistance}}
		\pgfpathlineto{\pgfpoint{\pgfdecoratedinputsegmentremainingdistance}{0pt}}
	}
}
\def\tornpaper#1{%
	\ifthenelse{\isodd{\value{mathseed}}}
	{%
		\tikz
		{
			\node[inner sep = 1em] (A) {#1};		
			\begin{pgfonlayer}{background}			
				\fill[paper]						
				\pgfextra{\pgfmathsetseed{\arabic{mathseed}}\addtocounter{mathseed}{1}}%
				{decorate[irregular cloudy border]{decorate{decorate{decorate{decorate[ragged border]{
										(A.north west) -- (A.north east)
				}}}}}}
				-- (A.south east)
				\pgfextra{\pgfmathsetseed{\arabic{mathseed}}}%
				{decorate[irregular spiky border]{decorate{decorate{decorate{decorate[ragged border]{
										-- (A.south west)
				}}}}}}
				-- (A.north west);
			\end{pgfonlayer}
		}
	}
	{%
		\tikz{
			\node[inner sep=1em] (A) {#1};  
			\begin{pgfonlayer}{background}  
				\fill[paper] 
				\pgfextra{\pgfmathsetseed{\arabic{mathseed}}\addtocounter{mathseed}{1}}%
				{decorate[irregular spiky border]{decorate{decorate{decorate{decorate[ragged border]{
										(A.north east) -- (A.north west)
				}}}}}}
				-- (A.south west)
				\pgfextra{\pgfmathsetseed{\arabic{mathseed}}}%
				{decorate[irregular cloudy border]{decorate{decorate{decorate{decorate[ragged border]{
										-- (A.south east)
				}}}}}}
				-- (A.north east);
		\end{pgfonlayer}}
	}
}
\definecolor{MyLightRed}{RGB}{244, 213, 245}
\definecolor{WordRed}{RGB}{255, 0, 102}
\definecolor{RedDarkLightest}{HTML}{ff0088}
\definecolor{RedDarkLight}{HTML}{ea005f}
\definecolor{RedPurple}{HTML}{aa007f}
\definecolor{Purple}{HTML}{911146}
\definecolor{PurpleDark}{RGB}{102, 0, 102}
\definecolor{WordLightGreen}{RGB}{140, 214, 192}
\definecolor{WordGreen}{RGB}{0, 176, 80}
\definecolor{GreenLightest}{HTML}{00ffa0}
\definecolor{GreenLighter1}{HTML}{00b383}
\definecolor{GreenLighter2}{HTML}{00aa7f}
\definecolor{GreenDark}{HTML}{225522}
\definecolor{GreenTeal}{HTML}{008080}
\definecolor{WordIceBlue}{RGB}{223, 227, 229}
\definecolor{MyVeryLightBlue}{RGB}{211, 245, 247}
\definecolor{WordBlueVeryLight}{RGB}{0, 176, 240}
\definecolor{WordBlueLight}{RGB}{0, 112, 192}
\definecolor{WordBlueDark}{RGB}{46, 116, 181}
\definecolor{WordBlueDarker}{RGB}{31, 78, 121}
\definecolor{WordBlueDarker25}{RGB}{54, 96, 146}
\definecolor{WordBlueDarker50}{RGB}{36, 64, 98}
\definecolor{WordBlueDarkest}{RGB}{0, 32, 96}
\definecolor{WordBlue}{RGB}{19, 65, 99}
\definecolor{MyBlue}{RGB}{0, 64, 128}
\definecolor{MyDarkBlue}{RGB}{0, 51, 102}
\definecolor{BlueVeryDark}{HTML}{222255}
\definecolor{MagentaVeryLight}{RGB}{178, 162, 201}
\definecolor{MagentaLighter}{RGB}{161, 106, 221}
\definecolor{MagentaLight}{RGB}{128, 100, 162}
\definecolor{MagentaDark}{RGB}{106, 65, 152}
\definecolor{MagentaVeryDark}{RGB}{97, 75, 128}
\definecolor{WordAquaLighter80}{RGB}{218, 238, 243}
\definecolor{WordAquaLighter60}{RGB}{183, 222, 232}
\definecolor{WordAquaLighter40}{RGB}{146, 205, 220}
\definecolor{WordAquaDarker25}{RGB}{49, 134, 155}
\definecolor{WordAquaDarker50}{RGB}{33, 89, 103}
\definecolor{WordVeryLightTeal}{RGB}{223, 236, 235}
\definecolor{WordLightTeal}{RGB}{160, 199, 197}
\definecolor{WordDarkTealLighter80}{RGB}{207, 223, 234}
\definecolor{WordDarkTeal}{RGB}{72, 123, 119}
\definecolor{WordDarkerTeal}{RGB}{48, 82, 80}
\definecolor{WordTurquoiseLighter80}{RGB}{209, 238, 249}
\definecolor{Brown}{HTML}{666633}

\usepackage[a4paper, left = 2.5cm, right = 2.5 cm, top = 2.5cm, bottom = 3.0 cm]{geometry}
\title
	{
		A Quantum Approach to News Verification from the Perspective of a News Aggregator
	}

\author
{
	Theodore Andronikos$^1$\orcidicon{0000-0002-3741-1271}
	and
	Alla Sirokofskich$^2$\\
	\\
	$^1$ \ Department of Informatics, Ionian University, \\
	7 Tsirigoti Square, 49100 Corfu, Greece; \\
	andronikos@ionio.gr \\
	$^2$ \ Department of History and Philosophy of Sciences, \\
	National and Kapodistrian University of Athens, \\
	Athens 15771, Greece; \\
	asirokof@math.uoa.gr
}
\begin{document}

\maketitle

\begin{abstract}
	In the dynamic landscape of digital information, the rise of misinformation and fake news presents a pressing challenge. This paper takes a completely new approach to verifying news, inspired by how quantum actors can reach agreement even when they are spatially spread out. We propose a radically new, to the best of our knowledge, algorithm that uses quantum ``entanglement'' (think of it as a special connection) to help news aggregators sniff out bad actors, whether they be other news sources or even fact-checkers trying to spread misinformation. This algorithm doesn't rely on quantum signatures, it just uses basic quantum technology we already have, in particular, special pairs of particles called ``EPR pairs'' that are much easier to create than other options. More complex entangled states are like juggling too many balls – they're hard to make and slow things down, especially when many players are involved. For instance, bigger, more complex states like ``GHZ states'' work for small groups, but they become messy with larger numbers. So, we stick with Bell states, the simplest form of entanglement, which are easy to generate no matter how many players are in the game. This means our algorithm is faster to set up, works for any number of participants, and is more practical for real-world use. Bonus points: it finishes in a fixed number of steps, regardless of how many players are involved, making it even more scalable. This new approach may lead to a powerful and efficient way to fight misinformation in the digital age, using the weird and wonderful world of quantum mechanics.
\end{abstract}
\section{Introduction} \label{sec: Introduction}

In the contemporary digital era, the proliferation of fake news, defined as deliberately false information masquerading as legitimate news, has emerged as a pervasive challenge across online and social media platforms. The rapid dissemination of misinformation poses serious repercussions, eroding trust in institutions, inciting violence, and undermining democratic processes. The urgent need for robust fake news detection mechanisms is more pronounced than ever. Fake news flourishes in the online realm due to several contributing factors. The accessibility of content creation and sharing, coupled with the absence of stringent oversight and anonymity, empowers malicious actors to disseminate false narratives with impunity. Furthermore, social media algorithms, often designed to prioritize sensational and engaging content, inadvertently amplify the reach of fake news.

The prevalence of fake news underscores the necessity for the development and deployment of effective detection techniques. One strategy involves leveraging fact-checking organizations to manually verify the veracity of information. However, this approach is labor-intensive and unable to cope with the sheer volume of content produced daily. Alternatively, employing artificial intelligence (AI) and machine learning (ML) methodologies offers promise in automatically identifying fake news. These algorithms can scrutinize various factors, including language usage, writing style, and source reliability, to flag potentially misleading content. Despite the potential of AI-driven detection methods, they encounter challenges. Fake news purveyors continuously adapt their strategies to evade detection, and AI models may exhibit biases or inaccuracies if trained on inadequate or skewed datasets. A paradox emerges as AI, while possessing the capability to identify and mitigate false news, simultaneously facilitates the proliferation of such deceptive online content. Notwithstanding these obstacles, the pursuit of effective fake news detection remains imperative. By combating the dissemination of misinformation, we safeguard individuals and society against its deleterious effects, nurturing a more informed, civil, and democratic online discourse.

Recent research (\cite{Campan2017}) underscores the necessity for social media platforms to integrate diverse content verification techniques alongside existing algorithms and AI approaches to combat false news effectively. Additionally, taxonomy frameworks like the one proposed in \cite{Vorhies2017, Shu2017}, which categorizes false news into distinct types, can aid social media platforms in alerting users to potentially misleading content, contingent upon agreed-upon standards for content analysis. The endeavor to automate the detection and prevention of false news presents formidable challenges, particularly concerning the assessment of content legitimacy (\cite{Rashkin2017, Mustafaraj2017}). Contemporary efforts predominantly rely on machine learning techniques to identify and mitigate fake news articles, as evidenced by numerous recent scholarly works (\cite{Gilda2017, Farajtabar2017, Agudelo2018, Kesarwani2020, Kesarwani2020a, Vijayaraghavan2020, Nagashri2021, Pandey2022}). The fusion of AI with blockchain technology emerges as a promising avenue for combating fake news (\cite{Tee2018}). This approach offers a decentralized and trustworthy platform for validating consent, authenticity, integrity, and perspectives on truth, thereby mitigating the spread of false narratives.

The quest for quantum computers that dethrone their classical counterparts continues. While we haven't reached the promised land yet, recent landmarks like IBM's 127-qubit Eagle \cite{IBMEagle}, 433-qubit Osprey \cite{IBMOsprey}, and the colossal 1,121-qubit Condor \cite{IBMCondor} show the path forward is accelerating. Perhaps we are closer than we think to the quantum revolution. Given this broader context, it becomes evident that quantum technology has reached a level of maturity where it merits serious consideration for inclusion in a comprehensive framework aimed at combating misinformation, especially considering the potential of quantum computers to enhance the speed and efficiency of Machine Learning algorithms. Researchers are actively pursuing the development of algorithmic methods that could effectively detect fake news and deepfakes by integrating Quantum Machine Learning techniques \cite{Kamal2022}. Quantum Machine Learning seeks to merge the principles of quantum computing with those of Machine Learning, offering tangible advantages such as improved Deep Fake detection \cite{QuantumIntegrity2019}. Tian et al. \cite{Tian2021} proposed a fake news detection system utilizing quantum K-Nearest Neighbors. Furthermore, Google has introduced an open-source library for Quantum Machine Learning, suggesting the potential for quantum computing to address fake and deepfake challenges in the near term \cite{Google2020}.

However, this paper explores a distinct quantum approach. It does not rely on Quantum Machine Learning, but, instead, draws inspiration from successful quantum protocols that achieve distributed consensus and detectable Byzantine Agreement in distributed environments (refer to recent work by Andronikos et al. \cite{Andronikos2023a} and related literature). We acknowledge the prevalent practice in contemporary social media platforms, wherein independent third-party fact-checking organizations, certified by impartial authorities, are employed to identify, assess, and take action on content. This fact-checking methodology primarily targets viral misinformation, particularly blatant falsehoods lacking factual basis. Ideally, fact-checking entities prioritize verifying demonstrably false claims that are both timely and impactful. Naturally, fact-checkers themselves should be subject to scrutiny and ongoing evaluation. The algorithm proposed herein envisions a decentralized setting where numerous news aggregators are overseen by multiple news verifiers responsible for content authentication. Described as a quantum game, our algorithm involves familiar figures such as Alice and Bob, alongside their numerous counterparts. Employing games in our presentation aims to facilitate comprehension of technical concepts. Quantum games, since their inception in 1999 \cite{Meyer1999, Eisert1999}, have garnered significant attention. The main reason for his trend is the potential superiority of quantum strategies over classical ones \cite{Andronikos2018, Andronikos2021, Andronikos2022a, Kastampolidou2023a}. Notably, the renowned prisoners' dilemma game exemplifies this phenomenon, extending to other abstract quantum games as well \cite{Eisert1999, Giannakis2019}. Moreover, the quantization of various classical systems can be applied to political structures, as demonstrated recently in \cite{Andronikos2022}. In the realm of quantum cryptography, the presentation of protocols often takes the form of games, a common practice evident in recent works such as \cite{Ampatzis2021, Ampatzis2022, Ampatzis2023, Andronikos2023, Andronikos2023a, Andronikos2023b}. Quantum strategies have demonstrated superiority over classical ones in various scenarios \cite{Andronikos2018, Andronikos2021, Andronikos2022a, Kastampolidou2023a}. The prisoners' dilemma game serves as a prominent example, and its applicability extends to other abstract quantum games \cite{Eisert1999, Giannakis2019}. Notably, the quantization of classical systems finds applications even in political structures \cite{Andronikos2022}. In the broader context of game-theoretic applications, unconventional environments, such as biological systems, have garnered significant attention \cite{Theocharopoulou2019, Kastampolidou2020a, Kostadimas2021}. It's intriguing to note that biological systems may give rise to biostrategies that outperform classical ones, even in iconic games like the Prisoners' Dilemma \cite{Kastampolidou2020, Kastampolidou2021, Kastampolidou2023, Papalitsas2021, Adam2023}.

\textbf{Contribution}. This paper introduces a novel perspective on the pressing issue of news verification by drawing inspiration from quantum protocols achieving distributed consensus, diverging from the more conventional Quantum Machine Learning approach. We present the first entanglement-based algorithm, to the best of our knowledge, designed for news aggregators to identify potential malicious actors. These actors could include other news aggregators or, even more concerning, fact-checkers intentionally disseminating misinformation.

The key advantage of our algorithm, which does not rely on a quantum signature scheme, lies in its compatibility with current quantum technology, as it solely depends on EPR pairs. While more complex multi-particle entangled states are possible, they are challenging to produce with existing quantum systems, leading to extended preparation times and complexity, particularly in scenarios involving numerous participants. For example, while contemporary quantum computers can easily generate $\ket{GHZ_n}$ states for small values of $n$, the preparation and distribution of these states become increasingly difficult as $n$ grows. Therefore, we exclusively employ Bell states, the simplest among maximally entangled states, in our algorithm.

Utilizing only EPR pairs, specifically $\ket{\Phi^+}$ pairs, regardless of the number of news aggregators and verifiers, results in reduced preparation time, improved scalability, and enhanced practicality. Additionally, our algorithm completes in a constant number of steps, irrespective of the number of participants, further enhancing its scalability and efficiency.

\subsection*{Organization} \label{subsec: Organization}

This article is organized as follows. The Introduction (Section \ref{sec: Introduction}) presents the subject matter, accompanied by bibliographic references to related works. A concise overview of the essential concepts is provided in Section \ref{sec: Background & Notation}, laying the foundation for understanding our protocol. A detailed explanation of the hypotheses underlying the QNVA is given in Section \ref{sec: Hypotheses & Setting}. The QNVA is formally presented in Section \ref{sec: The Quantum News Verification Algorithm}, explaining its inner workings in detail.
The paper concludes with a summary and discussion of the finer points of the algorithm in Section \ref{sec: Discussion and Conclusions}.

\section{Background \& notation} \label{sec: Background & Notation}

\subsection{EPR pairs} \label{subsec: EPR Pairs}

Quantum entanglement is a phenomenon where two or more particles become linked in such a way that they share the same fate, despite being separated by vast distances. This connection is so powerful that measuring the properties of one particle instantly determines the corresponding properties of its entangled partner, regardless of the separation between them. This instantaneous correlation defies our classical understanding of the universe and has profound implications for quantum mechanics and its potential applications. Mathematically, a single product state is not sufficient to describe entangled states of composite systems. So, they must be described as a linear combination of two or more product states of their subsystems. The famous Bell states are special quantum states of two qubits, also called EPR pairs, that represent the simplest form of maximal entanglement. These states can be compactly described by the next formula from \cite{Nielsen2010}.

\begin{align} \label{eq: Bell States General Equation}
	\ket{ \beta_{ x, y } } = \frac { \ket{ 0 } \ket{ y } + (-1)^x \ket{ 1 } \ket{ \overline{ y } } } { \sqrt{ 2 } } \ ,
\end{align}

where $\ket{ \overline{ y } }$ is the negation of $\ket{ y }$.

There are four Bell states and their specific mathematical expression is given below. The subscripts $A$ and $B$ are used to emphasize the subsystem to which the corresponding qubit belongs, that is, qubits $\ket{ \cdot }_{ A }$ belong to Alice and qubits $\ket{ \cdot }_{ B }$ belong to Bob.

\begin{tcolorbox}
	[
		grow to left by = 1.00 cm,
		grow to right by = 0.00 cm,
		colback = white,			
		enhanced jigsaw,			
		sharp corners,
		toprule = 0.1 pt,
		bottomrule = 0.1 pt,
		leftrule = 0.1 pt,
		rightrule = 0.1 pt,
		sharp corners,
		center title,
		fonttitle = \bfseries
	]
	\begin{minipage}[b]{0.475 \textwidth}
		\begin{align} \label{eq: Bell State Phi +}
			\ket{ \Phi^{ + } } = \ket{ \beta_{ 00 } } = \frac { \ket{ 0 }_{ A } \ket{ 0 }_{ B } + \ket{ 1 }_{ A } \ket{ 1 }_{ B } } { \sqrt{ 2 } }
		\end{align}
	\end{minipage} 
	\hfill
	\begin{minipage}[b]{0.45 \textwidth}
		\begin{align} \label{eq: Bell State Phi -}
			\ket{ \Phi^{ - } } = \ket{ \beta_{ 10 } } = \frac { \ket{ 0 }_{ A } \ket{ 0 }_{ B } - \ket{ 1 }_{ A } \ket{ 1 }_{ B } } { \sqrt{ 2 } }
		\end{align}
	\end{minipage}
	\begin{minipage}[b]{0.475 \textwidth}
		\begin{align} \label{eq: Bell State Psi +}
			\ket{ \Psi^{ + } } = \ket{ \beta_{ 01 } } = \frac { \ket{ 0 }_{ A } \ket{ 1 }_{ B } + \ket{ 1 }_{ A } \ket{ 0 }_{ B } } { \sqrt{ 2 } }
		\end{align}
	\end{minipage} 
	\hfill
	\begin{minipage}[b]{0.45 \textwidth}
		\begin{align} \label{eq: Bell State Psi -}
			\ket{ \Psi^{ - } } = \ket{ \beta_{ 11 } } = \frac { \ket{ 0 }_{ A } \ket{ 1 }_{ B } - \ket{ 1 }_{ A } \ket{ 0 }_{ B } } { \sqrt{ 2 } }
		\end{align}
	\end{minipage}
\end{tcolorbox}

For existing quantum computers based on the circuit model, it is quite trivial to produce Bell states. The proposed algorithm relies on $\ket{ \Phi^{ + } } = \frac { \ket{ 0 }_{ A } \ket{ 0 }_{ B } + \ket{ 1 }_{ A } \ket{ 1 }_{ B } } { \sqrt{ 2 } }$ pairs.

\subsection{The state $\ket{ + }$} \label{subsec: The States ket +}

Apart from $\ket{ \Phi^{ + } }$ pairs, our scheme makes use of another signature state, namely $\ket{ + }$. For completeness, we recall the definition of $\ket{ + }$ below

\begin{tcolorbox}
	[
		grow to left by = 0.00 cm,
		grow to right by = 0.00 cm,
		colback = white,			
		enhanced jigsaw,			
		sharp corners,
		toprule = 0.1 pt,
		bottomrule = 0.1 pt,
		leftrule = 0.1 pt,
		rightrule = 0.1 pt,
		sharp corners,
		center title,
		fonttitle = \bfseries
	]
		\begin{align} \label{eq: Ket +}
			\ket{ + } = H \ket{ 0 } = \frac { \ket{ 0 } + \ket{ 1 } } { \sqrt{ 2 } }
			\ .
		\end{align}
\end{tcolorbox}

State $\ket{ + }$ can be readily produced by applying the Hadamard transform on $\ket{ 0 }$. In the rest of this paper, the set of bit values $\{ 0, 1 \}$ is denoted by $\mathbb { B }$. As a final note, let us clarify that measurements are always performed with respect to the computational basis $\{ \ket{ 0 }, \ket{ 1 } \}$. 

\section{Hypotheses \& setting} \label{sec: Hypotheses & Setting}

Before we proceed with the comprehensive presentation of the proposed algorithm, it is useful to explicitly state the envisioned setting and the hypotheses that underlie the execution of our algorithm.

\subsection{The protagonists} \label{subsec: The Protagonists}

As we have previously mentioned, we follow what can, almost, be considered a tradition and describe the proposed algorithm as a game. The players in this game are the spatially distributed news verifiers and news aggregators, collectively called Alice and Bob. First, we state the actors that appear in our settings and clarify the roles they are supposed to play.

\begin{enumerate} [ left = 0.50 cm, labelsep = 0.75 cm, start = 1 ]
	\renewcommand\labelenumi{(\textbf{H}$_{ \theenumi }$)}
	\item	\emph{A trusted quantum source}. There exists a trusted quantum source that generates single qubits in the $\ket{ + }$ state and EPR pairs in the $\ket{ \Phi^{ + } }$ state. The source also distributes these qubits to all other players through appropriate quantum channels, according to the entanglement distribution scheme outlined in the forthcoming Definition \ref{def: Entanglement Distribution Scheme}.
	\item	\emph{News verifiers}. There exist $m$ special players that are called news verifiers. Their mission is to fact-check every piece of news and classify it as \emph{true} of \emph{fake}. In our game this role is undertaken by Alice and her clones, who are denoted by Alice$_{ 1 }$, \dots, Alice$_{ m }$. The news verifiers work independently of each other, and no communication, classical or quantum, takes place between any two of them.
	\item	\emph{News aggregators}. There are also $n$ players that are called news aggregators, and whose purpose is to gather and disseminate news that have been certified as true. This role is assumed by Bob and his clones that are denoted by Bob$_{ 1 }$, \dots, Bob$_{ n }$, where, typically, $n > m$.
\end{enumerate}

\subsection{The connections among the players} \label{subsec: The Connections Among the Players}

Besides the players listed above, there is a network of quantum and classical channels that enables the exchange of information among the players. In particular, we assume the existence of the following communication channels.

\begin{enumerate} [ left = 0.50 cm, labelsep = 0.75 cm, start = 4 ]
	\renewcommand\labelenumi{(\textbf{H}$_{ \theenumi }$)}
	\item	It is more realistic to consider that each Alice clone is not responsible for all news aggregators, but only for a specific group of news aggregators that are under her supervision. Henceforth, we shall assume that each Alice$_{ i }$, $1 \leq i \leq m$, is connected via pairwise authenticated classical channels to a specific subset of news aggregators who constitute her \emph{active network}, and Alice$_{ i }$ is their \emph{coordinator}. These aggregators are Alice$_{ i }$'s \emph{receivers}, their cardinality is denoted by $n_{ i }$ and they are designated by Bob$_{ 1 }^{ i }$, Bob$_{ 2 }^{ i }$, \dots, Bob$_{ n_{ i } }^{ i }$.
	\item	Each Alice$_{ i }$, $1 \leq i \leq m$, sends two things to every Bob$_{ k }^{ i }$, $1 \leq k \leq n_{ i }$, in her active network:
	\begin{itemize} [ left = 0.10 cm ]
		\item[$\Diamond$]	The result of the verification check, denoted by $c_{ k }^{ i }$.
		\item[$\Diamond$]	A \emph{proof sequence}, denoted by $\mathbf { p }_{ k }^{ i }$, which is intended to convince Bob$_{ k }^{ i }$ that she is honest.
	\end{itemize}
	The situation is visually depicted in Figure \ref{fig: Alice Sends Decision and Proof to her Bobs}.
	\begin{tcolorbox}
		[
			grow to left by = 0.00 cm,
			grow to right by = 0.00 cm,
			colback = MagentaLight!03,			
			enhanced jigsaw,					
			sharp corners,
			toprule = 1.0 pt,
			bottomrule = 1.0 pt,
			leftrule = 0.1 pt,
			rightrule = 0.1 pt,
			sharp corners,
			center title,
			fonttitle = \bfseries
		]
		\begin{figure}[H]
			\centering
			\begin{tikzpicture} [scale = 1.00]
				\def \n {8}
				\def \Angle {360 / \n}
				\def \Radius {4.75}
				\node
					[
						shade, top color = WordBlueDarker25, bottom color = black, rectangle, text width = 10.00 cm, align = center
					] ( Label ) at ( 0.0, 7.00 )
					{
						\color{white} Alice$_{ i }$ sends the verification outcome $c_{ k }^{ i }$ and the proof sequence $\mathbf { p }_{ k }^{ i }$ to every Bob$_{ k }^{ i }$ in her active network.
					};
				\draw [ line width = 2.00 pt, MyBlue ] ( 0, 0 ) circle [ radius = \Radius cm ];
				\node
					[
						dave,
						scale = 1.50,
						anchor = center,
						label = { [ label distance = 0.00 cm ] east: Bob$_{ 1 }^{ i }$ }
					]
					( ) at ( { \Radius * cos(1 * \Angle) }, { \Radius * sin(1 * \Angle) } ) { };
				\node
					[
						charlie,
						scale = 1.50,
						anchor = center,
						label = { [ label distance = 0.00 cm ] west: Bob$_{ 2 }^{ i }$ }
					]
					( ) at ( { \Radius * cos(3 * \Angle) }, { \Radius * sin(3 * \Angle) } ) { };
				\node [ shade, shading = ball, ball color = WordAquaLighter80, circle ] () at ( { \Radius * cos(4 * \Angle) }, { \Radius * sin(4 * \Angle) } ) {};
				\node [ shade, shading = ball, ball color = WordAquaLighter80, circle ] () at ( { \Radius * cos(5 * \Angle) }, { \Radius * sin(5 * \Angle) } ) {};
				\node [ shade, shading = ball, ball color = WordAquaLighter80, circle ] () at ( { \Radius * cos(6 * \Angle) }, { \Radius * sin(6 * \Angle) } ) {};
				\node
					[
						bob,
						scale = 1.50,
						anchor = center,
						label = { [ label distance = 0.00 cm ] east: Bob$_{ n_{ i } }^{ i }$ }
					]
					( ) at ( { \Radius * cos(7 * \Angle) }, { \Radius * sin(7 * \Angle) } ) { };
				\node
					[
						alice,
						scale = 1.50,
						anchor = center,
						label = { [ label distance = 0.00 cm ] south: Alice$_{ i }$ }
					]
					(Alice) at ( 0.00, 0.00 ) { };
				\begin{scope}[on background layer]
					\foreach \angle / \index in { 45/1, 135/2, 315/n_{ i } }
					\draw [ MyDarkBlue, ->, -{ Latex [ width = 14mm, length = 9mm ] }, line width = 6.25 mm, ] 
					( { 1.00 * cos(\angle) }, { 1.00 * sin(\angle) } ) --
					( { 4.25 * cos(\angle) }, { 4.25 * sin(\angle) } )
					node [ midway, text = white, rotate = \angle ] {$c_{ \index }^{ i }, \ \mathbf { p }_{ \index }^{ i }$};
					\foreach \angle in { 180, 225, 270 }
					\node [ shade, shading = ball, ball color = WordAquaLighter80, circle ] () at ( { \Radius * 0.65 * cos(\angle) }, { \Radius * 0.65 * sin(\angle) } ) {};
				\end{scope}
				\node [ anchor = center, below = 5.00 cm of Alice ] (PhantomNode) { };
			\end{tikzpicture}
			\caption{The above figure illustrates the fact that Alice$_{ i }$, $1 \leq i \leq m$, sends the verification outcome $c_{ k }^{ i }$ and the proof sequence $\mathbf { p }_{ k }^{ i }$ to every Bob$_{ k }^{ i }$, $1 \leq k \leq n_{ i }$, in her active network.} \label{fig: Alice Sends Decision and Proof to her Bobs}
		\end{figure}
	\end{tcolorbox}
	\item	All news aggregators that belong to the same active network are connected via pairwise authenticated classical channels. This enables them to exchange, whenever they deem necessary, the verification outcomes and the proof sequences they received from their coordinator. This action can be considered as an extra layer of verification and an indirect way in which aggregators can assess the honesty of other aggregators and also of the coordinator. This topology is shown in Figure \ref{fig: Bob Sends Decision and Proof to Other Bobs}. We clarify that aggregators that have no coordinator in common, do not communicate in any way.
	\begin{tcolorbox}
		[
			grow to left by = 0.00 cm,
			grow to right by = 0.00 cm,
			colback = MagentaLight!03,			
			enhanced jigsaw,					
			sharp corners,
			toprule = 1.0 pt,
			bottomrule = 1.0 pt,
			leftrule = 0.1 pt,
			rightrule = 0.1 pt,
			sharp corners,
			center title,
			fonttitle = \bfseries
		]
		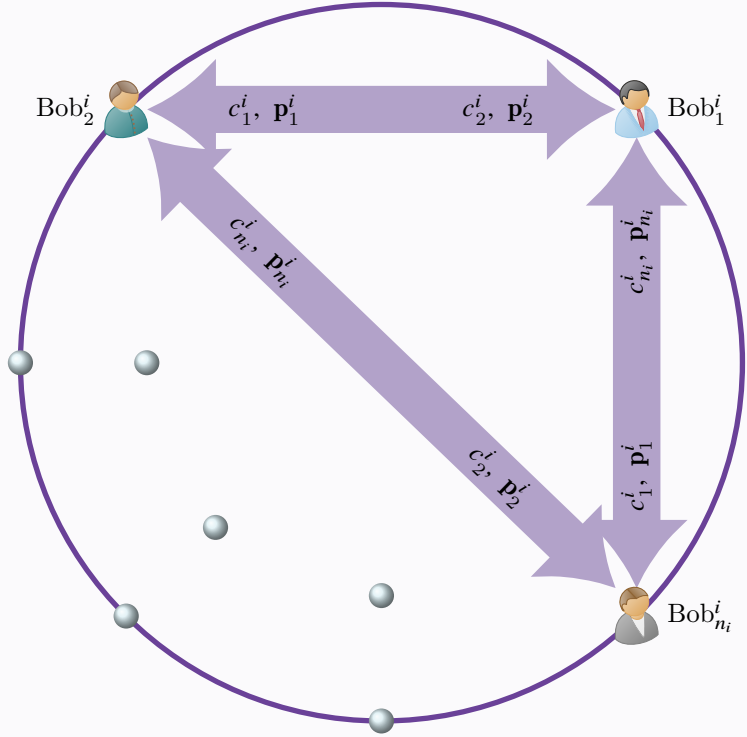
\begin{figure}[H]
			\centering
			\begin{tikzpicture} [scale = 1.00]
				\def \n {8}
				\def \Angle {360 / \n}
				\def \Radius {4.75}
				\node
					[
						shade, top color = WordAquaLighter80, bottom color = WordAquaDarker25, rectangle, text width = 11.00 cm, align = center
					] ( Label ) at ( 0.0, 7.00 )
					{
						\color{black} Aggregators Bob$_{ 1 }^{ i }$, Bob$_{ 2 }^{ i }$, \dots, Bob$_{ n_{ i } }^{ i }$ that have the same coordinator Alice$_{ i }$ exchange the verification outcomes and the proof sequences they received from Alice$_{ i }$ through pairwise authenticated classical channels.
					};
				\draw [ line width = 2.00 pt, MagentaDark ] ( 0, 0 ) circle [ radius = \Radius cm ];
				\node
					[
						dave,
						scale = 1.50,
						anchor = center,
						label = { [ label distance = 0.00 cm ] east: Bob$_{ 1 }^{ i }$ }
					]
					(Dave) at ( { \Radius * cos(1 * \Angle) }, { \Radius * sin(1 * \Angle) } ) { };
				\node
					[
						charlie,
						scale = 1.50,
						anchor = center,
						label = { [ label distance = 0.00 cm ] west: Bob$_{ 2 }^{ i }$ }
					]
					(Charlie) at ( { \Radius * cos(3 * \Angle) }, { \Radius * sin(3 * \Angle) } ) { };
				\node [ shade, shading = ball, ball color = WordAquaLighter80, circle ] () at ( { \Radius * cos(4 * \Angle) }, { \Radius * sin(4 * \Angle) } ) {};
				\node [ shade, shading = ball, ball color = WordAquaLighter80, circle ] () at ( { \Radius * cos(5 * \Angle) }, { \Radius * sin(5 * \Angle) } ) {};
				\node [ shade, shading = ball, ball color = WordAquaLighter80, circle ] () at ( { \Radius * cos(6 * \Angle) }, { \Radius * sin(6 * \Angle) } ) {};
					\node
					[
						bob,
						scale = 1.50,
						anchor = center,
						label = { [ label distance = 0.00 cm ] east: Bob$_{ n_{ i } }^{ i }$ }
					]
					(Bob) at ( { \Radius * cos(7 * \Angle) }, { \Radius * sin(7 * \Angle) } ) { };
				\begin{scope}[on background layer]
					\draw [ MagentaVeryLight, <->, { Latex [ width = 14mm, length = 9mm ] }-{ Latex [ width = 14mm, length = 9mm ] }, line width = 6.25 mm, ] 
					(Charlie.east) -- (Dave.west)
					node [ black, near start, sloped ] { $c_{ 1 }^{ i }, \ \mathbf { p }_{ 1 }^{ i }$ }
					node [ black, near end, sloped ] { $c_{ 2 }^{ i }, \ \mathbf { p }_{ 2 }^{ i }$ };
					\draw [ MagentaVeryLight, <->, { Latex [ width = 14mm, length = 9mm ] }-{ Latex [ width = 14mm, length = 9mm ] }, line width = 6.25 mm, ]
					(Bob.north west) -- (Charlie.south east)
					node [ black, near start, sloped ] { $c_{ 2 }^{ i }, \ \mathbf { p }_{ 2 }^{ i }$ }
					node [ black, near end, sloped ] { $c_{ n_{ i } }^{ i }, \ \mathbf { p }_{ n_{ i } }^{ i }$ };
					\draw [ MagentaVeryLight, <->, { Latex [ width = 14mm, length = 9mm ] }-{ Latex [ width = 14mm, length = 9mm ] }, line width = 6.25 mm, ]
					(Bob.north) -- (Dave.south)
					node [ black, near start, sloped ] { $c_{ 1 }^{ i }, \ \mathbf { p }_{ 1 }^{ i }$ }
					node [ black, near end, sloped ] { $c_{ n_{ i } }^{ i }, \ \mathbf { p }_{ n_{ i } }^{ i }$ };
					\foreach \angle in { 180, 225, 270 }
					\node [ shade, shading = ball, ball color = WordAquaLighter80, circle ] () at ( { \Radius * 0.65 * cos(\angle) }, { \Radius * 0.65 * sin(\angle) } ) {};
				\end{scope}
				\node [ anchor = center, below = 5.00 cm of Alice ] (PhantomNode) { };
			\end{tikzpicture}
			\caption{The above figure illustrates the fact that all news aggregators with the same coordinator Alice$_{ i }$, i.e., Bob$_{ 1 }^{ i }$, Bob$_{ 2 }^{ i }$, \dots, Bob$_{ n_{ i } }^{ i }$, can exchange the verification outcomes and the proof sequences they received from Alice$_{ i }$ through pairwise authenticated classical channels.} \label{fig: Bob Sends Decision and Proof to Other Bobs}
		\end{figure}
	\end{tcolorbox}
	\item	Every news aggregator is responsible for maintaining the reputation system outlined below, independently, and in parallel with every other news aggregator.
	\begin{itemize} [ left = 0.10 cm ]
		\item[$\Diamond$]	A news ranking system that characterizes news as either true or fake.
		\item[$\Diamond$]	A reputation catalog that reflects the personal assessment of the aggregator regarding every other player (both verifier and aggregator) involved in information exchange.
	\end{itemize}
	News deemed as fake must be appropriately flagged as such, so that the public is made aware of this. The reputation catalog takes the form of two lists containing the unreliable verifiers and aggregators.
\end{enumerate}

The intuition behind the latter hypothesis
quite straightforward. It is expedient to record and consider unreliable those players that have exhibited contradictory and/or malicious behavior, and distinguish them from those players that have demonstrated consistent adherence to the rules and have a history of accurate and truthful reporting. By maintaining these records, each aggregator plays an important role in ensuring the integrity and efficiency of the news network. By identifying and isolating unreliable entities, he helps to protect the network from malicious actors and maintain the trust among participants.

\section{The quantum news verification algorithm} \label{sec: The Quantum News Verification Algorithm}

In this section we present the quantum news verification algorithm, or QNVA for short. Every Alice is tasked with verifying important news, and sending the result of her verification check to all her agents.

\begin{itemize}
	\item[$\Diamond$]	If the news in question passed the verification check, then Alice sends via the classical channel the bit $1$ to every Bob in her active network to signify its credibility. Additionally, she sends a \emph{personalized proof}, which is a sequence of bits, to each of her agents. The important thing here is that for each Bob the proof is different because it is constructed specifically for him.
	\item[$\Diamond$]	Symmetrically, if the news failed to pass the check, Alice sends via the classical channel the bit $0$ to every agent in her active network to indicate that it is fake, together with a personalized proof.
\end{itemize}

The QNVA is presented from the point of view of the individual Bob, where, of course, we assume that all Bobs implement the same algorithm consistently. The algorithm itself can be conceptually organized into $3$ distinct phases.

\subsection{The entanglement distribution phase} \label{subsec: The Entanglement Distribution Phase}

The first is the \textbf{entanglement distribution phase}, which refers to the generation and distribution of entangled $\ket{ \Phi^{ + } }$ pairs and qubits in the $\ket{ + }$ state. As we have explained in hypothesis ($\mathbf { H }_{ 1 }$), we assume the existence of a trusted quantum source that undertakes this task. In view of the capabilities of modern quantum apparatus, the trusted quantum source should have no difficulty in accomplishing this task. In terms of notation, we use the small Latin letters $q$ and $r$, with appropriate subscripts and superscripts, to designate the first and the second qubit, respectively, of the same $\ket{ \Phi^{ + } }$ pair.

\begin{definition} [Entanglement Distribution Scheme] \label{def: Entanglement Distribution Scheme}
	The trusted source creates two types of sequences, one that is intended for verifiers and one that is intended for aggregators. Both are distributed through quantum channels to their intended recipients. Specifically, for each piece of news that must be checked, and for each Alice$_{ i }$, $1 \leq i \leq m$, the source creates
	\begin{itemize}
		\item[$\Diamond$]	one \emph{verification} sequence $\mathbf { q }^{ i }$ that is sent to Alice$_{ i }$ and has the form
		\begin{align}
			\mathbf { q }^{ i }
			=
			\underbrace { \colorbox {WordAquaLighter40} { $q_{ n_{ i }, d }^{ i } \dots q_{ k, d }^{ i } \dots q_{ 1, d }^{ i }$ } }_{ \text{ tuple } d }
			\cdots
			\underbrace { \colorbox {WordAquaLighter60} { $q_{ n_{ i }, 2 }^{ i } \dots q_{ k, 2 }^{ i } \dots q_{ 1, 2 }^{ i }$ } }_{ \text{ tuple } 2 }
			\underbrace { \colorbox {WordAquaLighter80} { $q_{ n_{ i }, 1 }^{ i } \dots q_{ k, 1 }^{ i } \dots q_{ 1, 1 }^{ i }$ } }_{ \text{ tuple } 1 }
			\ , \text{ and }
			\label{eq: Alice's Verification Sequence}
		\end{align}
		\item[$\Diamond$]	$n_{ i }$ \emph{verification} sequences $\mathbf { r }_{ 1 }^{ i }$, $\mathbf { r }_{ 2 }^{ i }$, \dots, $\mathbf { r }_{ n_{ i } }^{ i }$ sent to Bob$_{ 1 }^{ i }$, Bob$_{ 2 }^{ i }$, \dots, Bob$_{ n_{ i } }^{ i }$, respectively, that have the form
		\begin{align}
			\mathbf { r }_{ k }^{ i }
			=
			\underbrace { \colorbox {WordLightGreen} { $\ket{ + } \dots r_{ k, d }^{ i } \dots \ket{ + }$ } }_{ \text{ tuple } d }
			\cdots
			\underbrace { \colorbox {WordLightGreen!66} { $\ket{ + } \dots r_{ k, 2 }^{ i } \dots \ket{ + }$ } }_{ \text{ tuple } 2 } \
			\underbrace { \colorbox {WordLightGreen!33} { $\ket{ + } \dots r_{ k, 1 }^{ i } \dots \ket{ + }$ } }_{ \text{ tuple } 1 }
			\ , \ 1 \leq k \leq n_{ i }
			\ .
			\label{eq: Bob's Verification Sequence}
		\end{align}
	\end{itemize}
	In the $\mathbf { q }^{ i }$ and $\mathbf { r }_{ 1 }^{ i }$, $\mathbf { r }_{ 2 }^{ i }$, \dots, $\mathbf { r }_{ n_{ i } }^{ i }$ sequences, the subscript $d$ is a positive integer called the \emph{degree of accuracy} of the verification. Furthermore, according to our convention, $q_{ k, l }^{ i }$ and $r_{ k, l }^{ i }$ denote the first and second qubits of the same $\ket{ \Phi^{ + } }$ pair that is used in the formation of the $l^{ th }$ tuple. Obviously, $\ket{ + }$ designates qubits that are in the $\ket{ + }$ state.
\end{definition}

The situation regarding the sequences of qubits distributed to Alice$_{ i }$ and the Bobs in her active network is visualized in Figure \ref{fig: Alice's and her Bobs' Quantum Sequences}.

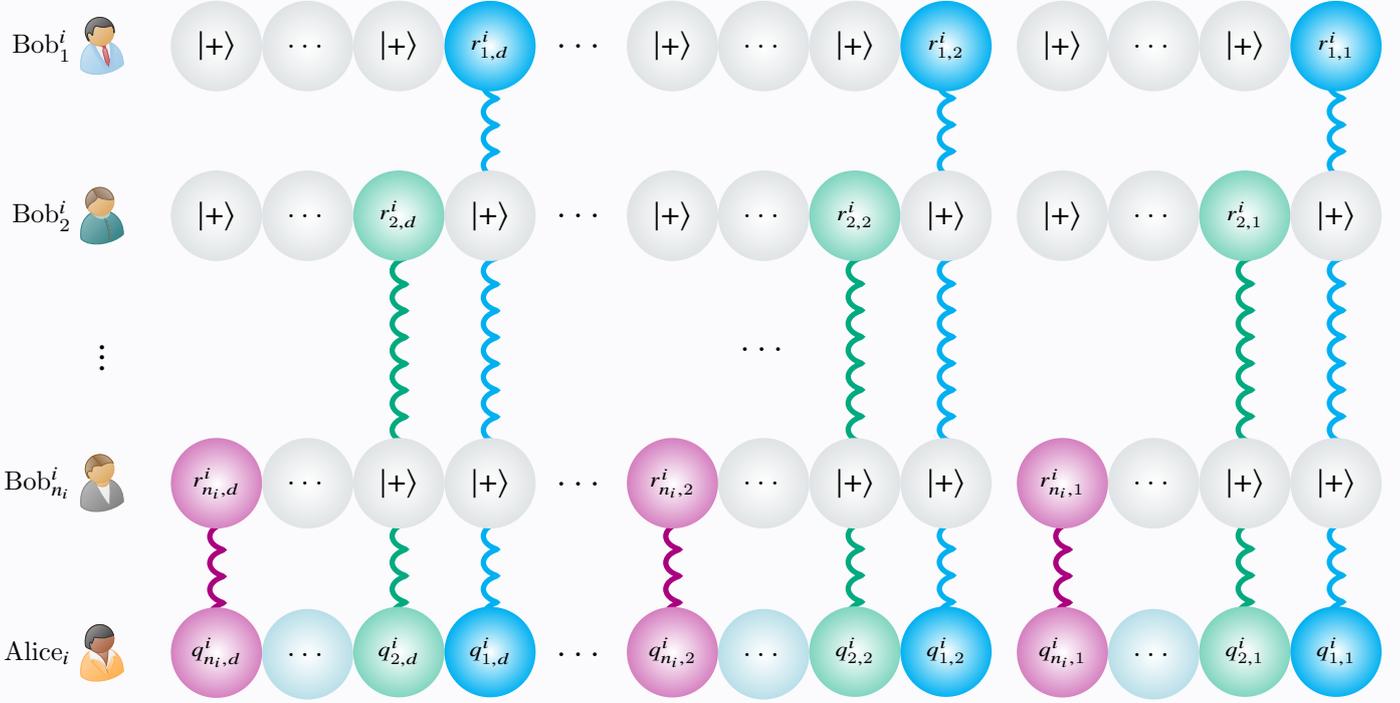
\begin{figure}[H]
	\begin{tcolorbox}
		[
			grow to left by = 1.50 cm,
			grow to right by = 1.50 cm,
			colback = MagentaLight!03,				
			enhanced jigsaw,						
			sharp corners,
			toprule = 1.0 pt,
			bottomrule = 1.0 pt,
			leftrule = 0.1 pt,
			rightrule = 0.1 pt,
			sharp corners,
			center title,
			fonttitle = \bfseries,
		]
		\centering
		\begin{tikzpicture} [ scale = 0.25 ]
			\node
				[
					alice,
					scale = 1.50,
					anchor = center,
					label = { [ label distance = 0.00 cm ] west: Alice$_{ i }$ }
				]
				(Alice) { };
			\matrix
				[
				matrix of nodes, nodes in empty cells,
				column sep = 0.000 mm, right = 0.50 of Alice,
				nodes = { circle, minimum size = 12 mm, semithick, font = \footnotesize },
				]
				{
					\node [ shade, outer color = RedPurple!50, inner color = white ] (A-n-d) { $q_{ n_{ i }, d }^{ i }$ }; &
					\node [ shade, outer color = WordAquaLighter60, inner color = white ] { \large \dots }; &
					\node [ shade, outer color = GreenLighter2!50, inner color = white ] (A-2-d) { $q_{ 2, d }^{ i }$ }; &
					\node [ shade, outer color = WordBlueVeryLight, inner color = white ] (A-1-d) { $q_{ 1, d }^{ i }$ }; &
					\node { \Large \dots }; &
					\node [ shade, outer color = RedPurple!50, inner color = white ] (A-n-2) { $q_{ n_{ i }, 2 }^{ i }$ }; &
					\node [ shade, outer color = WordAquaLighter60, inner color = white ] { \large \dots }; &
					\node [ shade, outer color = GreenLighter2!50, inner color = white ] (A-2-2) { $q_{ 2, 2 }^{ i }$ }; &
					\node [ shade, outer color = WordBlueVeryLight, inner color = white ] (A-1-2) { $q_{ 1, 2 }^{ i }$ }; &
					\node [ minimum size = 0.1 mm ] { }; &
					\node [ shade, outer color = RedPurple!50, inner color = white ] (A-n-1) { $q_{ n_{ i }, 1 }^{ i }$ }; &
					\node [ shade, outer color = WordAquaLighter60, inner color = white ] { \large \dots }; &
					\node [ shade, outer color = GreenLighter2!50, inner color = white ] (A-2-1) { $q_{ 2, 1 }^{ i }$ }; &
					\node [ shade, outer color = WordBlueVeryLight, inner color = white ] (A-1-1) { $q_{ 1, 1 }^{ i }$ };
					\\
				};
			\node
				[
					bob,
					scale = 1.50,
					anchor = center,
					above = 1.50 cm of Alice,
					label = { [ label distance = 0.00 cm ] west: Bob$_{ n_{ i } }^{ i }$ }
				]
				(Bob) { };
				\matrix
				[
					column sep = 0.000 mm, right = 0.50 of Bob,
					nodes = { circle, minimum size = 12 mm, semithick, font = \footnotesize },
				]
				{
					\node [ shade, outer color = RedPurple!50, inner color = white ] (B-n-d) { $r_{ n_{ i }, d }^{ i }$ }; &
					\node [ shade, outer color = WordIceBlue, inner color = white ] { \large \dots }; &
					\node [ shade, outer color = WordIceBlue, inner color = white ] { \large $\ket{ + }$ }; &
					\node [ shade, outer color = WordIceBlue, inner color = white ] { \large $\ket{ + }$ }; &
					\node { \Large \dots }; &
					\node [ shade, outer color = RedPurple!50, inner color = white ] (B-n-2) { $r_{ n_{ i }, 2 }^{ i }$ }; &
					\node [ shade, outer color = WordIceBlue, inner color = white ] { \large \dots }; &
					\node [ shade, outer color = WordIceBlue, inner color = white ] { \large $\ket{ + }$ }; &
					\node [ shade, outer color = WordIceBlue, inner color = white ] { \large $\ket{ + }$ }; &
					\node [ minimum size = 0.1 mm ] { }; &
					\node [ shade, outer color = RedPurple!50, inner color = white ] (B-n-1) { $r_{ n_{ i }, 1 }^{ i }$ }; &
					\node [ shade, outer color = WordIceBlue, inner color = white ] { \large \dots }; &
					\node [ shade, outer color = WordIceBlue, inner color = white ] { \large $\ket{ + }$ }; &
					\node [ shade, outer color = WordIceBlue, inner color = white ] { \large $\ket{ + }$ };
					\\
				};
			\node
				[
					anchor = center,
					above = 1.00 cm of Bob,
				]
				(Dots) { \Large \vdots };
			\matrix
				[
					column sep = 0.000 mm, right = 3.00 of Dots,
					nodes = { circle, minimum size = 12 mm, semithick, font = \footnotesize },
				]
				(DotsMatrix)
				{
					\node { }; &
					\node { }; &
					\node { }; &
					\node { }; &
					\node { \Large \dots }; &
					\node { };
					\\
				};
			\node
				[
					charlie,
					scale = 1.50,
					anchor = center,
					above = 1.00 cm of Dots,
					label = { [ label distance = 0.00 cm ] west: Bob$_{ 2 }^{ i }$ }
				]
				(Charlie) { };
			\matrix
				[
					column sep = 0.000 mm, right = 0.50 of Charlie,
					nodes = { circle, minimum size = 12 mm, semithick, font = \footnotesize },
				]
				{
					\node [ shade, outer color = WordIceBlue, inner color = white ] { \large $\ket{ + }$ }; &
					\node [ shade, outer color = WordIceBlue, inner color = white ] { \large \dots }; &
					\node [ shade, outer color = GreenLighter2!50, inner color = white ] (C-2-d) { $r_{ 2, d }^{ i }$ }; &
					\node [ shade, outer color = WordIceBlue, inner color = white ] { \large $\ket{ + }$ }; &
					\node { \Large \dots }; &
					\node [ shade, outer color = WordIceBlue, inner color = white ] { \large $\ket{ + }$ }; &
					\node [ shade, outer color = WordIceBlue, inner color = white ] { \large \dots }; &
					\node [ shade, outer color = GreenLighter2!50, inner color = white ] (C-2-2) { $r_{ 2, 2 }^{ i }$ }; &
					\node [ shade, outer color = WordIceBlue, inner color = white ] { \large $\ket{ + }$ }; &
					\node [ minimum size = 0.1 mm ] { }; &
					\node [ shade, outer color = WordIceBlue, inner color = white ] { \large $\ket{ + }$ }; &
					\node [ shade, outer color = WordIceBlue, inner color = white ] { \large \dots }; &
					\node [ shade, outer color = GreenLighter2!50, inner color = white ] (C-2-1) { $r_{ 2, 1 }^{ i }$ }; &
					\node [ shade, outer color = WordIceBlue, inner color = white ] { \large $\ket{ + }$ };
					\\
				};
			\node
				[
					dave,
					scale = 1.50,
					anchor = center,
					above = 1.50 cm of Charlie,
					label = { [ label distance = 0.00 cm ] west: Bob$_{ 1 }^{ i }$ }
				]
				(Dave) { };
			\matrix
				[
					column sep = 0.000 mm, right = 0.50 of Dave,
					nodes = { circle, minimum size = 12 mm, semithick, font = \footnotesize },
				]
				{
					\node [ shade, outer color = WordIceBlue, inner color = white ] { \large $\ket{ + }$ }; &
					\node [ shade, outer color = WordIceBlue, inner color = white ] { \large \dots }; &
					\node [ shade, outer color = WordIceBlue, inner color = white ] { \large $\ket{ + }$ }; &
					\node [ shade, outer color = WordBlueVeryLight, inner color = white ] (D-1-d) { $r_{ 1, d }^{ i }$ }; &
					\node { \Large \dots }; &
					\node [ shade, outer color = WordIceBlue, inner color = white ] { \large $\ket{ + }$ }; &
					\node [ shade, outer color = WordIceBlue, inner color = white ] { \large \dots }; &
					\node [ shade, outer color = WordIceBlue, inner color = white ] { \large $\ket{ + }$ }; &
					\node [ shade, outer color = WordBlueVeryLight, inner color = white ] (D-1-2) { $r_{ 1, 2 }^{ i }$ }; &
					\node [ minimum size = 0.1 mm ] { }; &
					\node [ shade, outer color = WordIceBlue, inner color = white ] { \large $\ket{ + }$ }; &
					\node [ shade, outer color = WordIceBlue, inner color = white ] { \large \dots }; &
					\node [ shade, outer color = WordIceBlue, inner color = white ] { \large $\ket{ + }$ }; &
					\node [ shade, outer color = WordBlueVeryLight, inner color = white ] (D-1-1) { $r_{ 1, 1 }^{ i }$ };
					\\
				};
			\begin{scope}[on background layer]
				\draw [ RedPurple, -, >=stealth, line width = 0.7 mm, decoration = coil, decorate ]
				(A-n-d.center) -- (B-n-d.center);
				\draw [ RedPurple, -, >=stealth, line width = 0.7 mm, decoration = coil, decorate ]
				(A-n-2.center) -- (B-n-2.center);
				\draw [ RedPurple, -, >=stealth, line width = 0.7 mm, decoration = coil, decorate ]
				(A-n-1.center) -- (B-n-1.center);
				\draw [ GreenLighter2, -, >=stealth, line width = 0.7 mm, decoration = coil, decorate ]
				(A-2-d.center) -- (C-2-d.center);
				\draw [ GreenLighter2, -, >=stealth, line width = 0.7 mm, decoration = coil, decorate ]
				(A-2-2.center) -- (C-2-2.center);
				\draw [ GreenLighter2, -, >=stealth, line width = 0.7 mm, decoration = coil, decorate ]
				(A-2-1.center) -- (C-2-1.center);
				\draw [ WordBlueVeryLight, -, >=stealth, line width = 0.7mm , decoration = coil, decorate ]
				(A-1-d.center) -- (D-1-d.center);
				\draw [ WordBlueVeryLight, -, >=stealth, line width = 0.7mm , decoration = coil, decorate ]
				(A-1-2.center) -- (D-1-2.center);
				\draw [ WordBlueVeryLight, -, >=stealth, line width = 0.7mm , decoration = coil, decorate ]
				(A-1-1.center) -- (D-1-1.center);
			\end{scope}
			\node
				[
				above right = 2.00 cm and 7.50 cm of Dave, anchor = center, shade, top color = MagentaLighter, bottom color = black, rectangle, text width = 9.50 cm, align = center
				]
				(Label)
				{ \color{white} \textbf{The entangled sequences of qubits distributed to Alice and the news aggregators in her active network.} };
			\node [ anchor = west, below = 0.50 cm of Alice ] (PhantomNode1) { };
			\node [ anchor = west, above = 0.25 cm of Label ] (PhantomNode2) { };
		\end{tikzpicture}
	\end{tcolorbox}
	\caption{The above figure depicts the entangled sequences of qubits distributed to Alice and the news aggregators in her active network. Qubits that belong to the same $\ket{ \Phi^{ + } }$ pair are indicated by the same color and a wavy line that connects them. Specifically, blue indicates the EPR pairs shared between Alice$_{ i }$ and Bob$_{ 1 }^{ i }$, which occupy position $1$ in each $n_{ i }$-tuple of the $\mathbf { q }^{ i }$ and $\mathbf { r }_{ 1 }^{ i }$ sequences. Analogously, green is used for the EPR pairs shared between Alice$_{ i }$ and Bob$_{ 2 }^{ i }$, and red for the EPR pairs linking Alice$_{ i }$ and Bob$_{ n_{ i } }^{ i }$. The silver qubits designate those in the $\ket{ + }$ state that fill the remaining positions of the sequences $\mathbf { r }_{ 1 }^{ i }$, $\mathbf { r }_{ 2 }^{ i }$, \dots, $\mathbf { r }_{ n_{ i } }^{ i }$.
	}
	\label{fig: Alice's and her Bobs' Quantum Sequences}
\end{figure}

The visual convention is to draw qubits that belong to the same $\ket{ \Phi^{ + } }$ pair with the same color. Blue is used to indicate the EPR pairs shared between Alice$_{ i }$ and Bob$_{ 1 }^{ i }$, which occupy position $1$ in each $n_{ i }$-tuple of the $\mathbf { q }^{ i }$ and $\mathbf { r }_{ 1 }^{ i }$ sequences. Analogously, green is used for the EPR pairs shared between Alice$_{ i }$ and Bob$_{ 2 }^{ i }$ located in the second position of every $n_{ i }$-tuple of the $\mathbf { q }^{ i }$ and $\mathbf { r }_{ 2 }^{ i }$ sequences, and red signifies EPR pairs linking Alice$_{ i }$ and Bob$_{ n_{ i } }^{ i }$ occupying the last position of every $n_{ i }$-tuple of the $\mathbf { q }^{ i }$ and $\mathbf { r }_{ n_{ i } }^{ i }$ sequences. The silver qubits designate those in the $\ket{ + }$ state that fill the remaining positions of the sequences $\mathbf { r }_{ 1 }^{ i }$, $\mathbf { r }_{ 2 }^{ i }$, \dots, $\mathbf { r }_{ n_{ i } }^{ i }$. The intuition behind the construction of the above quantum sequences is outlined below.

\begin{enumerate}  [ left = 0.50 cm, labelsep = 1.00 cm, start = 1 ]
	\renewcommand\labelenumi{(\textbf{I}$_{ \theenumi }$)}
	\item	Alice$_{ i }$, $1 \leq i \leq m$, is linked to each one of her agents Bob$_{ 1 }^{ i }$, Bob$_{ 2 }^{ i }$, \dots, Bob$_{ n_{ i } }^{ i }$ because her verification sequence $\mathbf { q }^{ i }$ is entangled with their verification sequences sequences $\mathbf { r }_{ 1 }^{ i }$, $\mathbf { r }_{ 2 }^{ i }$, \dots, $\mathbf { r }_{ n_{ i } }^{ i }$.
	\item	All these quantum sequences are made up of $d$ in total $n_{ i }$-tuples of qubits.
	\item	Sequence $\mathbf { q }^{ i }$ is made up exclusively from entangled qubits.
	\item	In $\mathbf { q }^{ i }$ the qubits in position $1$, namely $q_{ 1, 1 }^{ i }$, $q_{ 1, 2 }^{ i }$, \dots, $q_{ 1, d }^{ i }$, are entangled with the corresponding qubits $r_{ 1, 1 }^{ i }$, $r_{ 1, 2 }^{ i }$, \dots, $r_{ 1, d }^{ i }$ of the sequence $\mathbf { r }_{ 1 }^{ i }$ that belongs to Bob$_{ 1 }^{ i }$. This is because $q_{ 1, l }^{ i }$ and $r_{ 1, l }^{ i }$, $1 \leq l \leq d$, belong to the same $\ket{ \Phi^{ + } }$ pair by construction.
	\item	For precisely the same reason, the qubits in position $k, \ k = 2, \dots, n_{ i }$, i.e., $q_{ k, 1 }^{ i }$, $q_{ k, 2 }^{ i }$, \dots, $q_{ k, d }^{ i }$, are entangled with the corresponding qubits $r_{ k, 1 }^{ i }$, $r_{ k, 2 }^{ i }$, \dots, $r_{ k, d }^{ i }$ of the sequence $\mathbf { r }_{ k }^{ i }$ owned by Bob$_{ k }^{ i }$.
	\item	In every sequence $\mathbf { r }_{ k }^{ i }$, $k = 1, \dots, n_{ i }$, the qubits $r_{ k, l }^{ i }$, $l = 1, \dots, d$, that occupy the $k^{ th }$ position in each $n_{ i }$-tuple, are entangled with the corresponding qubits $q_{ k, l }^{ i }$ of $\mathbf { q }^{ i }$. All other qubits are in the $\ket{ + }$ state.
\end{enumerate}

In Section \ref{sec: Discussion and Conclusions}, where we discuss the effect of the degree of accuracy of the QNVA, we shall suggest appropriate values for $d$.

In view of the structural form of the sequences defined by formulae \eqref{eq: Alice's Verification Sequence} and \eqref{eq: Bob's Verification Sequence}, we also refer to them as $( d, n_{ i } )$ quantum sequences because they are constructed by $d$ repetitions of structurally similar tuples of the same length, namely $n_{ i }$. These $d$ tuples are enumerated as shown in \eqref{eq: Alice's Verification Sequence} and \eqref{eq: Bob's Verification Sequence}, that is tuple $1$ is the rightmost tuple and tuple $d$ is the leftmost tuple.

\subsection{Entanglement validation phase} \label{subsec: Entanglement Validation Phase}

Undoubtedly, this is a most crucial phase, as the entire algorithm hinges upon the existence of entanglement. Without guaranteed entanglement, the algorithm's functionality is compromised. The validation procedure can result into two distinct outcomes. If entanglement is successfully ascertained, the QNVA can proceed to confidently verify the information at hand. Failure to validate entanglement indicates the absence of the necessary entanglement. This could stem from either noisy quantum channels or malicious interference from an adversary. Regardless of the cause, the only viable solution is to halt the ongoing algorithm execution and commence the entire procedure anew, after implementing corrective measures.

Given its utmost significance, this phase has undergone thorough scrutiny in the existing literature. Our algorithm adheres to the sophisticated methodologies outlined in prior works, including \cite{Neigovzen2008, Feng2019, Wang2022a, Yang2022, Qu2023, Ikeda2023c}. Hence, to preclude redundant exposition, we direct the reader to the previously mentioned bibliography for all the details essential for the successful implementation of this phase.

\subsection{The news verification phase} \label{subsec: The News Verification Phase}

Our algorithm classifies news as true or fake during the third and last phase, aptly named \textbf{news verification phase}. To initiate this phase, Alice$_{ i }$ and the agents in her active network, Bob$_{ 1 }^{ i }$, Bob$_{ 2 }^{ i }$, \dots, Bob$_{ n_{ i } }^{ i }$, measure their quantum sequences to obtain the classical bit sequences denoted by the lower case bold letters $\mathbf { a }^{ i }$ and $\mathbf { b }_{ 1 }^{ i }$, $\mathbf { b }_{ 2 }^{ i }$, \dots, $\mathbf { b }_{ n_{ i } }^{ i }$, respectively. Taking into account the entanglement distribution scheme of Definition \ref{def: Entanglement Distribution Scheme}, we see that the measurement has revealed some important correlations among these sequences. These correlations are depicted in the next Figure \ref{fig: Alice and her Bobs' Classical Bit Sequences}.

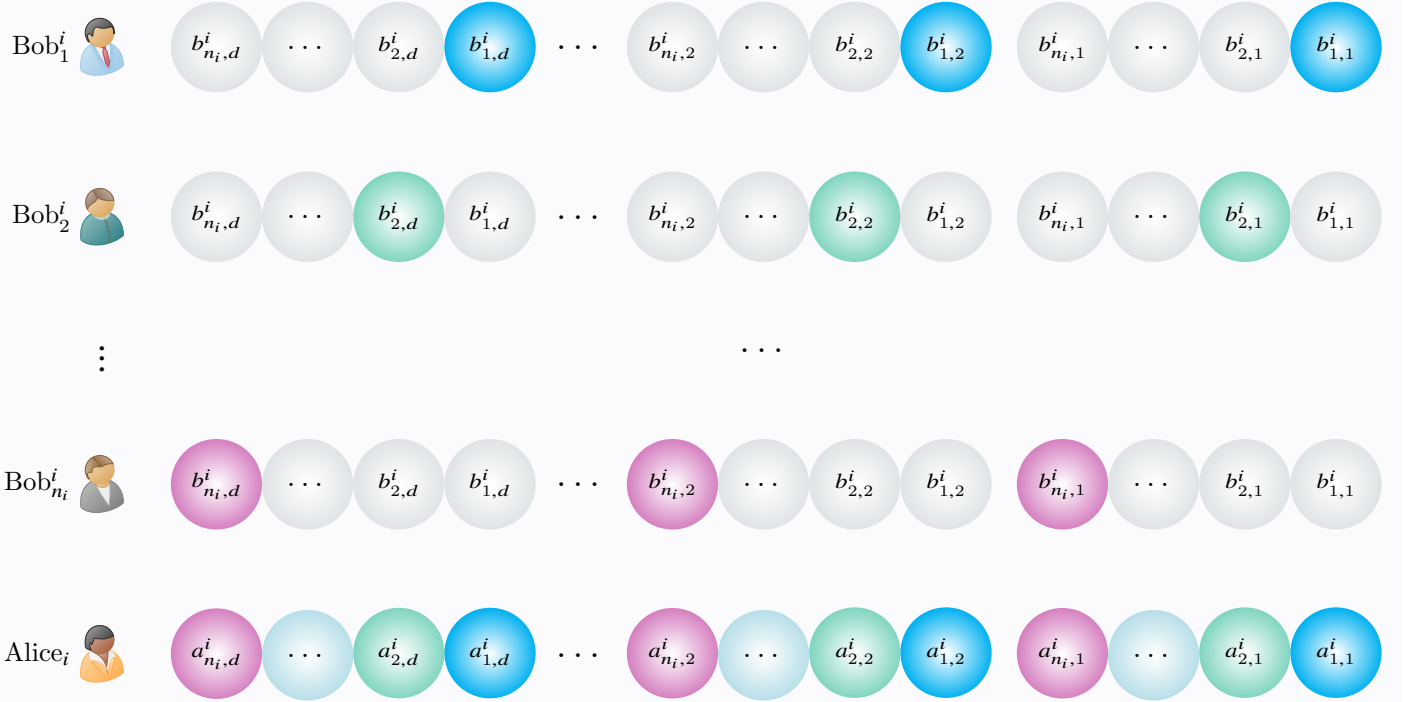
\begin{figure}[H]
	\begin{tcolorbox}
		[
			grow to left by = 1.50 cm,
			grow to right by = 1.50 cm,
			colback = MagentaLight!03,				
			enhanced jigsaw,						
			sharp corners,
			toprule = 1.0 pt,
			bottomrule = 1.0 pt,
			leftrule = 0.1 pt,
			rightrule = 0.1 pt,
			sharp corners,
			center title,
			fonttitle = \bfseries
		]
		\centering
		\begin{tikzpicture} [ scale = 0.25 ]
			\node
				[
					alice,
					scale = 1.50,
					anchor = center,
					label = { [ label distance = 0.00 cm ] west: Alice$_{ i }$ }
				]
				(Alice) { };
			\matrix
				[
					matrix of nodes, nodes in empty cells,
					column sep = 0.000 mm, right = 0.50 of Alice,
					nodes = { circle, minimum size = 12 mm, semithick, font = \footnotesize },
				]
				{
					\node [ shade, outer color = RedPurple!50, inner color = white ] (A-n-d) { $a_{ n_{ i }, d }^{ i }$ }; &
					\node [ shade, outer color = WordAquaLighter60, inner color = white ] { \large \dots }; &
					\node [ shade, outer color = GreenLighter2!50, inner color = white ] (A-2-d) { $a_{ 2, d }^{ i }$ }; &
					\node [ shade, outer color = WordBlueVeryLight, inner color = white ] (A-1-d) { $a_{ 1, d }^{ i }$ }; &
					\node { \Large \dots }; &
					\node [ shade, outer color = RedPurple!50, inner color = white ] (A-n-2) { $a_{ n_{ i }, 2 }^{ i }$ }; &
					\node [ shade, outer color = WordAquaLighter60, inner color = white ] { \large \dots }; &
					\node [ shade, outer color = GreenLighter2!50, inner color = white ] (A-2-2) { $a_{ 2, 2 }^{ i }$ }; &
					\node [ shade, outer color = WordBlueVeryLight, inner color = white ] (A-1-2) { $a_{ 1, 2 }^{ i }$ }; &
					\node [ minimum size = 0.1 mm ] { }; &
					\node [ shade, outer color = RedPurple!50, inner color = white ] (A-n-1) { $a_{ n_{ i }, 1 }^{ i }$ }; &
					\node [ shade, outer color = WordAquaLighter60, inner color = white ] { \large \dots }; &
					\node [ shade, outer color = GreenLighter2!50, inner color = white ] (A-2-1) { $a_{ 2, 1 }^{ i }$ }; &
					\node [ shade, outer color = WordBlueVeryLight, inner color = white ] (A-1-1) { $a_{ 1, 1 }^{ i }$ };
					\\
				};
			\node
				[
					bob,
					scale = 1.50,
					anchor = center,
					above = 1.50 cm of Alice,
					label = { [ label distance = 0.00 cm ] west: Bob$_{ n_{ i } }^{ i }$ }
				]
				(Bob) { };
			\matrix
				[
					column sep = 0.000 mm, right = 0.50 of Bob,
					nodes = { circle, minimum size = 12 mm, semithick, font = \footnotesize },
				]
				{
					\node [ shade, outer color = RedPurple!50, inner color = white ] (B-n-d) { $b_{ n_{ i }, d }^{ i }$ }; &
					\node [ shade, outer color = WordIceBlue, inner color = white ] { \large \dots }; &
					\node [ shade, outer color = WordIceBlue, inner color = white ] { $b_{ 2, d }^{ i }$ }; &
					\node [ shade, outer color = WordIceBlue, inner color = white ] { $b_{ 1, d }^{ i }$ }; &
					\node { \Large \dots }; &
					\node [ shade, outer color = RedPurple!50, inner color = white ] (B-n-2) { $b_{ n_{ i }, 2 }^{ i }$ }; &
					\node [ shade, outer color = WordIceBlue, inner color = white ] { \large \dots }; &
					\node [ shade, outer color = WordIceBlue, inner color = white ] { $b_{ 2, 2 }^{ i }$ }; &
					\node [ shade, outer color = WordIceBlue, inner color = white ] { $b_{ 1, 2 }^{ i }$ }; &
					\node [ minimum size = 0.1 mm ] { }; &
					\node [ shade, outer color = RedPurple!50, inner color = white ] (B-n-1) { $b_{ n_{ i }, 1 }^{ i }$ }; &
					\node [ shade, outer color = WordIceBlue, inner color = white ] { \large \dots }; &
					\node [ shade, outer color = WordIceBlue, inner color = white ] { $b_{ 2, 1 }^{ i }$ }; &
					\node [ shade, outer color = WordIceBlue, inner color = white ] { $b_{ 1, 1 }^{ i }$ };
					\\
				};
			\node
				[
					anchor = center,
					above = 1.00 cm of Bob,
				]
				(Dots) { \Large \vdots };
			\matrix
				[
					column sep = 0.000 mm, right = 3.00 of Dots,
					nodes = { circle, minimum size = 12 mm, semithick, font = \footnotesize },
				]
				(DotsMatrix)
				{
					\node { }; &
					\node { }; &
					\node { }; &
					\node { }; &
					\node { \Large \dots }; &
					\node { };
					\\
				};
			\node
				[
					charlie,
					scale = 1.50,
					anchor = center,
					above = 1.00 cm of Dots,
					label = { [ label distance = 0.00 cm ] west: Bob$_{ 2 }^{ i }$ }
				]
				(Charlie) { };
			\matrix
				[
					column sep = 0.000 mm, right = 0.50 of Charlie,
					nodes = { circle, minimum size = 12 mm, semithick, font = \footnotesize },
				]
				{
					\node [ shade, outer color = WordIceBlue, inner color = white ] { $b_{ n_{ i }, d }^{ i }$ }; &
					\node [ shade, outer color = WordIceBlue, inner color = white ] { \large \dots }; &
					\node [ shade, outer color = GreenLighter2!50, inner color = white ] (C-2-d) { $b_{ 2, d }^{ i }$ }; &
					\node [ shade, outer color = WordIceBlue, inner color = white ] { $b_{ 1, d }^{ i }$ }; &
					\node { \Large \dots }; &
					\node [ shade, outer color = WordIceBlue, inner color = white ] { $b_{ n_{ i }, 2 }^{ i }$ }; &
					\node [ shade, outer color = WordIceBlue, inner color = white ] { \large \dots }; &
					\node [ shade, outer color = GreenLighter2!50, inner color = white ] (C-2-2) { $b_{ 2, 2 }^{ i }$ }; &
					\node [ shade, outer color = WordIceBlue, inner color = white ] { $b_{ 1, 2 }^{ i }$ }; &
					\node [ minimum size = 0.1 mm ] { }; &
					\node [ shade, outer color = WordIceBlue, inner color = white ] { $b_{ n_{ i }, 1 }^{ i }$ }; &
					\node [ shade, outer color = WordIceBlue, inner color = white ] { \large \dots }; &
					\node [ shade, outer color = GreenLighter2!50, inner color = white ] (C-2-1) { $b_{ 2, 1 }^{ i }$ }; &
					\node [ shade, outer color = WordIceBlue, inner color = white ] { $b_{ 1, 1 }^{ i }$ };
					\\
				};
			\node
				[
					dave,
					scale = 1.50,
					anchor = center,
					above = 1.50 cm of Charlie,
					label = { [ label distance = 0.00 cm ] west: Bob$_{ 1 }^{ i }$ }
				]
				(Dave) { };
			\matrix
				[
					column sep = 0.000 mm, right = 0.50 of Dave,
					nodes = { circle, minimum size = 12 mm, semithick, font = \footnotesize },
				]
				{
					\node [ shade, outer color = WordIceBlue, inner color = white ] { $b_{ n_{ i }, d }^{ i }$ }; &
					\node [ shade, outer color = WordIceBlue, inner color = white ] { \large \dots }; &
					\node [ shade, outer color = WordIceBlue, inner color = white ] { $b_{ 2, d }^{ i }$ }; &
					\node [ shade, outer color = WordBlueVeryLight, inner color = white ] (D-1-d) { $b_{ 1, d }^{ i }$ }; &
					\node { \Large \dots }; &
					\node [ shade, outer color = WordIceBlue, inner color = white ] { $b_{ n_{ i }, 2 }^{ i }$ }; &
					\node [ shade, outer color = WordIceBlue, inner color = white ] { \large \dots }; &
					\node [ shade, outer color = WordIceBlue, inner color = white ] { $b_{ 2, 2 }^{ i }$ }; &
					\node [ shade, outer color = WordBlueVeryLight, inner color = white ] (D-1-2) { $b_{ 1, 2 }^{ i }$ }; &
					\node [ minimum size = 0.1 mm ] { }; &
					\node [ shade, outer color = WordIceBlue, inner color = white ] { $b_{ n_{ i }, 1 }^{ i }$ }; &
					\node [ shade, outer color = WordIceBlue, inner color = white ] { \large \dots }; &
					\node [ shade, outer color = WordIceBlue, inner color = white ] { $b_{ 2, 1 }^{ i }$ }; &
					\node [ shade, outer color = WordBlueVeryLight, inner color = white ] (D-1-1) { $b_{ 1, 1 }^{ i }$ };
					\\
				};
			\node
				[
					above right = 2.00 cm and 8.00 cm of Dave, anchor = center, shade, top color = WordDarkTeal, bottom color = black, rectangle, text width = 10.00 cm, align = center
				]
				(Label)
				{ \color{white} \textbf{The classical bit sequences resulting after the measurement of the quantum sequences and their correlations.} };
			\node [ anchor = west, below = 0.50 cm of Alice ] (PhantomNode1) { };
			\node [ anchor = west, above = 0.25 cm of Label ] (PhantomNode2) { };
		\end{tikzpicture}
	\end{tcolorbox}
	\caption{This figure shows the classical bit sequences that result after Alice and the news aggregators measure their quantum sequences. The correlations among pairs of bits in these sequences are visualized by drawing correlated pairs with the same color.
	}
	\label{fig: Alice and her Bobs' Classical Bit Sequences}
\end{figure}

\begin{definition} [Classical Bit Sequences] \label{def: Classical Bit Sequences}
	Upon measuring their qubit sequences, news verifier Alice$_{ i }$ and news aggregators Bob$_{ 1 }^{ i }$, Bob$_{ 2 }^{ i }$, \dots, Bob$_{ n_{ i } }^{ i }$, obtain the classical bit sequences $\mathbf { a }^{ i }$ and $\mathbf { b }_{ 1 }^{ i }$, $\mathbf { b }_{ 2 }^{ i }$, \dots, $\mathbf { b }_{ n_{ i } }^{ i }$, respectively, that can be written explicitly as follows.
	\begin{align}
		\mathbf { a }^{ i }
		&=
		\underbrace { \colorbox {WordAquaLighter40} { $a_{ n_{ i }, d }^{ i } \dots a_{ k, d }^{ i } \dots a_{ 1, d }^{ i }$ } }_{ \text{ tuple } d }
		\cdots
		\underbrace { \colorbox {WordAquaLighter60} { $a_{ n_{ i }, 2 }^{ i } \dots a_{ k, 2 }^{ i } \dots a_{ 1, 2 }^{ i }$ } }_{ \text{ tuple } 2 }
		\underbrace { \colorbox {WordAquaLighter80} { $a_{ n_{ i }, 1 }^{ i } \dots a_{ k, 1 }^{ i } \dots a_{ 1, 1 }^{ i }$ } }_{ \text{ tuple } 1 }
		\ , \text{ and }
		\label{eq: Alice's Classical Bit Sequence}
		\\
		\mathbf { b }_{ k }^{ i }
		&=
		\underbrace { \colorbox {WordLightGreen} { $b_{ n_{ i }, d }^{ i } \dots b_{ k, d }^{ i } \dots b_{ 1, d }^{ i }$ } }_{ \text{ tuple } d }
		\cdots
		\underbrace { \colorbox {WordLightGreen!50} { $b_{ n_{ i }, 2 }^{ i } \dots b_{ k, 2 }^{ i } \dots b_{ 1, 2 }^{ i }$ } }_{ \text{ tuple } 2 } \
		\underbrace { \colorbox {WordLightGreen!25} { $b_{ n_{ i }, 1 }^{ i } \dots b_{ k, 1 }^{ i } \dots b_{ 1, 1 }^{ i }$ } }_{ \text{ tuple } 1 }
		\label{eq: Bob's Classical Bit Sequence}
		\ ,
	\end{align}
	where $k = 1, \dots, n_{ i }$.
\end{definition}

Although the sequences defined by formulae \eqref{eq: Alice's Classical Bit Sequence} and \eqref{eq: Bob's Classical Bit Sequence} consist of classical bits, their structural form is identical to that of the original qubit sequences. So, we will also refer to them as $( d, n_{ i } )$ classical sequences because they are constructed by repeating $d$ times structurally similar tuples of length $n_{ i }$. Following this line of thought, we may consider an arbitrary $( d, n_{ i } )$ sequence $\mathbf { s }$ made of symbols from some fixed alphabet, and express it in an alternative but equivalent way, emphasizing its composition in terms of $n_{ i }$-tuples, as follows.

\begin{align}
	\mathbf { s }
	&=
	\underbrace { \colorbox {WordAquaLighter40} { $s_{ n_{ i }, d } \dots s_{ 2, d } s_{ 1, d }$ } }_{ \text{ tuple } d }
	\cdots
	\underbrace { \colorbox {WordAquaLighter60} { $s_{ n_{ i }, 2 } \dots s_{ 2, 2 } s_{ 1, 2 }$ } }_{ \text{ tuple } 2 }
	\underbrace { \colorbox {WordAquaLighter80} { $s_{ n_{ i }, 1 } \dots s_{ 2, 1 } s_{ 1, 1 }$ } }_{ \text{ tuple } 1 }
	\nonumber \\
	\mathbf { s }
	&=
	\hspace{ 1.05 cm }
	\overset { \downarrow } { \colorbox {MagentaVeryLight} { $\mathbf { s }_{ d }$ } }
	\hspace{ 1.00 cm }
	\cdots
	\hspace{ 1.05 cm }
	\overset { \downarrow } { \colorbox {MagentaVeryLight!40!MyLightRed} { $\mathbf { s }_{ 2 }$ } }
	\hspace{ 2.00 cm }
	\overset { \downarrow } { \colorbox {MyLightRed} { $\mathbf { s }_{ 1 }$ } }
	\label{eq: Classical Bit Sequence Tuple Form}
	\ .
\end{align}

Let us also clarify that by writing $\mathbf { s } = \mathbf { s }_{ d } \cdots \mathbf { s }_{ 2 } \mathbf { s }_{ 1 }$, where $\mathbf { s }_{ l } = s_{ n_{ i }, l } \dots s_{ 2, l } s_{ 1, l }$ and $1 \leq l \leq d$, we have effectively enumerated the $d$ tuples of $\mathbf { s }$ in a way that $1$ is the rightmost and $d$ is the leftmost tuple. In the rest of our exposition, we will also need a special $n_{ i }$-tuple that is constructed by using a new symbol $\ast$, different from $0$ and $1$. This is denoted by $\mathbf { s }_{ \ast }$ and is referred to the \emph{cryptic} tuple.

\begin{align}
	\mathbf { s }_{ \ast }
	=
	\underbrace { \colorbox {orange!25} { $\ast \ \dots \ \ast \ \ast$ } }_{ n_{ i } \text{ occurrences} }
	\label{eq: The Special Tuple}
	\ .
\end{align}

With the above convention, we may write Alice$_{ i }$'s bit sequence $\mathbf { a }^{ i }$ in the next form, emphasizing the fact that it is composed by $d$ tuples.

\begin{align}
	\mathbf { a }^{ i }
	=
	\colorbox {WordAquaLighter40} { $\mathbf { a }_{ d }$ }
	\cdots
	\colorbox {WordAquaLighter60} { $\mathbf { a }_{ 2 }$ }
	\colorbox {WordAquaLighter80} { $\mathbf { a }_{ 1 }$ }
	\label{eq: Alice's Classical Bit Sequence Tuple Form}
\end{align}

\begin{definition} [Proof Sequences] \label{def: Proof Sequences} \
	News verifier Alice$_{ i }$, $1 \leq i \leq m$, sends two things to all the news aggregators in her active network, namely Bob$_{ 1 }^{ i }$, Bob$_{ 2 }^{ i }$, \dots, Bob$_{ n_{ i } }^{ i }$:
	\begin{itemize}
		\item[$\Diamond$]	The result of the verification check, denoted by $c_{ k }^{ i } \in \mathbb { B }$, which is just a single bit. If the news is true, then $c_{ k }^{ i }$ is just the bit $1$, whereas if the news is fake, $c_{ k }^{ i }$ is the bit $0$.
		\item[$\Diamond$]	A \emph{proof sequence}, denoted by $\mathbf { p }_{ k }^{ i }$, which is intended to convince Bob$_{ k }^{ i }$ that she is honest. Each proof sequence $\mathbf { p }_{ k }^{ i }$ is a $( d, n_{ i } )$ sequence of symbols from $\mathbb { B } \cup \{ \ast \}$, i.e., $\mathbf { p }_{ k }^{ i } = \mathbf { p }_{ d } \ \dots \ \mathbf { p }_{ 2 } \ \mathbf { p }_{ 1 }$. It is critical that these proof sequences be personalized, which effectively means they must be different for every news aggregator. Their construction is described below.
		\begin{itemize}
			\item	If $c_{ k }^{ i } = 1$, the proof $\mathbf { p }_{ k }^{ i }$ sequence sent to news aggregator Bob$_{ k }^{ i }$, $1 \leq k \leq n_{ i }$, also designated by $\mathds{ 1 }_{ k }^{ i }$ for emphasis, has the explicit form shown below.
			\begin{align}
				\mathds{ 1 }_{ k }^{ i }
				&=
				\mathbf { p }_{ d } \ \dots \ \mathbf { p }_{ 2 } \ \mathbf { p }_{ 1 }
				\ ,
				\ \text{ where } \
				\mathbf { p }_{ l }
				=
				\left\{
				\begin{matrix*}[l]
					\ \mathbf { a }_{ l } & \text{ if } a_{ k, l }^{ i } = 1 \\
					\ \mathbf { s }_{ \ast } & \text{ if } a_{ k, l }^{ i } = 0
				\end{matrix*}
				\right.
				\ , \ 1 \leq l \leq d \ .
				\label{eq: Bob's k Proof Sequence for True}
			\end{align}
			\item	Symmetrically, if $c_{ k }^{ i } = 0$, the proof sequence sent to Bob$_{ k }^{ i }$, $1 \leq k \leq n_{ i }$, denoted by $\vmathbb{ 0 }_{ k }^{ i }$ for emphasis, has the following explicit form.
			\begin{align}
				\vmathbb{ 0 }_{ k }^{ i }
				&=
				\mathbf { p }_{ d } \ \dots \ \mathbf { p }_{ 2 } \ \mathbf { p }_{ 1 }
				\ ,
				\ \text{ where } \
				\mathbf { p }_{ l }
				=
				\left\{
				\begin{matrix*}[l]
					\ \mathbf { a }_{ l } & \text{ if } a_{ k, l }^{ i } = 0 \\
					\ \mathbf { s }_{ \ast } & \text{ if } a_{ k, l }^{ i } = 1
				\end{matrix*}
				\right.
				\ , \ 1 \leq l \leq d \ .
				\label{eq: Bob's k Proof Sequence for Fake}
			\end{align}
		\end{itemize}
	\end{itemize}
\end{definition}

A proof sequence for a verification check $c_{ k }^{ i }$ that is faithfully constructed according to Definition \ref{def: Proof Sequences} is said to be \emph{consistent} with $c_{ k }^{ i }$. The previous Definition \ref{def: Proof Sequences} guarantees that, no matter what the verification outcome is, Bob$_{ k }^{ i }$ receives a different proof sequence from every other Bob$_{ k^{ \prime } }^{ i }$ when $k^{ \prime } \neq k$. The other crucial property that characterizes the proof sequences is the fact that besides tuples comprised entirely of $0$ and $1$ bits, they also contain a statistically equal number of cryptic tuples consisting of the special symbol $\ast$. Probabilistic analysis allows Bob$_{ k }^{ i }$ to assess whether or not Alice$_{ i }$ and the other Bobs act honestly and consistently. To this end, it is necessary to examine the positions of the $n_{ i }$-tuples that contain specific combinations of bits.

\begin{definition} \label{def: Tuple Designation} \
	Let
	$
	\mathbf { s }
	=
	\underbrace { s_{ n_{ i }, d } \dots s_{ 2, d } s_{ 1, d } }_{ \text{ tuple } d }
	\cdots
	\underbrace { s_{ n_{ i }, 2 } \dots s_{ 2, 2 } s_{ 1, 2 } }_{ \text{ tuple } 2 }
	\underbrace { s_{ n_{ i }, 1 } \dots s_{ 2, 1 } s_{ 1, 1 } }_{ \text{ tuple } 1 }
	$
	be a $( d, n_{ i } )$ sequence. If $k$ and $k^{ \prime }$, $1 \leq k \neq k^{ \prime } \leq n_{ i }$, are two given indices, and $x, y \in \mathbb { B }$, we define the following sets
	%
	%
	\begin{align}
		P_{ x } ( \mathbf { s }, k )
		&=
		\{ l \ \vert \ s_{ k, l }  = x \}
		\ ,
		\label{eq: P_x Definition}
		\\
		P_{ x, y } ( \mathbf { s }, k, k^{ \prime } )
		&=
		\{ l \ \vert \ s_{ k, l }  = x \wedge s_{ k^{ \prime }, l }  = y \}
		\ , \text{ and }
		\label{eq: P_xy Definition}
		\\
		P_{ \ast } ( \mathbf { s } )
		&=
		\{ l \ \vert \ \mathbf { s }_{ l } = \mathbf { s }_{ \ast } \}
		\ .
		\label{eq: P_* Definition}
	\end{align}
\end{definition}

The previous definition completes the necessary machinery and notation for the presentation of the QNVA. We may now proceed to explain the QNVA in detail and, at the same time, prove its correctness. In the rest of this section, we present the algorithm from the point of view of the typical news aggregator Bob$_{ k }^{ i }$, $1 \leq k \leq n_{ i }$. In what follows, we use the notation $\lvert S \rvert$ to designate the number of elements of a given set $S$.

In today's complex news environment malicious intent can manifest in many subtle ways. One may easily envision the next most critical scenarios.

\begin{enumerate}  [ left = 0.00 cm, labelsep = 1.00 cm, start = 1 ]
	\renewcommand\labelenumi{(\textbf{S}$_{ \theenumi }$)}
	\item	An unreliable and dishonest Alice$_{ i }$ sends to Bob$_{ k }^{ i }$ the verification outcome $c_{ k }^{ i }$, but the latter is accompanied with the wrong proof sequence $\mathbf { p }_{ k }^{ i }$.
	\item	A malicious news aggregator, say Bob$_{ k^{ \prime } }^{ i }$ ($k^{ \prime } \neq k$), falsely claims that he received from Alice$_{ i }$ the opposite verification outcome accompanied by a consistent proof sequence.
	\item	An insidious Alice$_{ i }$ deliberately spreads disinformation and confusion by sending opposite verification outcomes $c_{ k }^{ i }$ and $\overline { c_{ k }^{ i } }$ to Bob$_{ k }^{ i }$ and Bob$_{ k^{ \prime } }^{ i }$ ($k^{ \prime } \neq k$), using consistent proof sequences $\mathbf { p }_{ k }^{ i }$ and $\mathbf { p }_{ k^{ \prime } }^{ i }$ in each case.
\end{enumerate}

The first scenario (\textbf{S}$_{ 1 }$) can be easily detected and countered by the QNVA. The second scenario (\textbf{S}$_{ 2 }$) can also be countered with additional effort. Our algorithm can also deal with the third scenario (\textbf{S}$_{ 3 }$), which reveals the existence of a truly malicious Alice, with some additional inference on the part of Bob. QNVA owes its ability to cope with each one of the above scenarios to the structural properties of the proof sequences. These properties are a direct result of the entanglement distribution scheme explained in Definition \ref{def: Entanglement Distribution Scheme}. The next Proposition \ref{prp: Is Alice's Proof Consistent?} lays the groundwork for the subsequent analysis of our verification procedures.

\begin{proposition} [Dishonest news verifier detection]  \label{prp: Is Alice's Proof Consistent?} \
	Let us assume that news aggregator Bob$_{ k }^{ i }$, $1 \leq k \leq n_{ i }$, has received from his coordinator Alice$_{ i }$, $1 \leq i \leq m$, the verification outcome $c_{ k }^{ i } \in \mathbb { B }$, and the proof sequence $\mathbf { p }_{ k }^{ i }$. If the proof sequence $\mathbf { p }_{ k }^{ i }$ is consistent with $c_{ k }^{ i }$, then it must satisfy the following properties.
	\begin{align}
		\mathbb { E }
		\left[
		\
		\lvert
		\
		P_{ c_{ k }^{ i } }
		( \mathbf { p }_{ k }^{ i }, k )
		\
		\rvert
		\
		\right]
		=
		\mathbb { E }
		\left[
		\
		\lvert
		\
		P_{ \ast }
		( \mathbf { p }_{ k }^{ i } )
		\
		\rvert
		\
		\right]
		&=
		\frac { d } { 2 }
		\ , \ \text{and}
		\label{eq: Expected Number of Tuples with Specific Value in Position k}
		\\
		\mathbb { E }
		\left[
		\
		\lvert
		\
		P_{ c_{ k }^{ i }, c_{ k }^{ i } }
		( \mathbf { p }_{ k }^{ i }, k, k^{ \prime }  )
		\
		\rvert
		\
		\right]
		=
		\mathbb { E }
		\left[
		\
		\lvert
		\
		P_{ c_{ k }^{ i }, \overline { c_{ k }^{ i } } }
		( \mathbf { p }_{ k }^{ i }, k, k^{ \prime }  )
		\
		\rvert
		\
		\right]
		&=
		\frac { d } { 4 }
		\ , \ \forall k^{ \prime } \ , \ 1 \leq k^{ \prime } \neq k \leq n_{ i } \ .
		\label{eq: Expected Number of Tuples with Specific Values in Positions k & k'}
	\end{align}
\end{proposition}

\begin{proof}
	The proof is quite straightforward because it is based on the entanglement distribution scheme of Definition \ref{def: Entanglement Distribution Scheme}. The entanglement distribution scheme stipulates that each $n_{ i }$-tuple in the original qubit sequence of Alice$_{ i }$ shares one $\ket{ \Phi^{ + } } = \frac { \ket{ 0 }_{ A } \ket{ 0 }_{ k } + \ket{ 1 }_{ A } \ket{ 1 }_{ k } } { \sqrt{ 2 } }$ pair with every Bob$_{ k }^{ i }$, $k = 1, \dots, n_{ i }$. Therefore, there are $d$ in total bits $a_{ k, l }^{ i }$ that occupy the $k^{ th }$ position in every tuple $l$ of $\mathbf { a }^{ i }$, $1 \leq l \leq d$, which are equal to the corresponding bits $b_{ k, l }^{ i }$ of $\mathbf { b }_{ k }^{ i }$. This is captured by the next formula:
	\begin{align}
		a_{ k, l }^{ i } = b_{ k, l }^{ i }
		\ , \ 1 \leq l \leq d \ .
		\label{eq: Alice & Bob's k Equal Bits}
	\end{align}
	The remaining bits of $\mathbf { b }_{ k }^{ i }$ result from measuring qubits in the $\ket{ + }$ state. Hence, we expect half of them to end up $0$, and the remaining half to end up $1$. More importantly though, these bits are not correlated with the corresponding bits of $\mathbf { a }^{ i }$. Consequently, we can easily draw the following conclusions.
	\begin{itemize}
		\item	Measuring a pair of qubits in the $\ket{ \Phi^{ + } }$ state will result in both qubits collapsing in state $\ket{ 0 }$ with probability $0.5$, or in state $\ket{ 1 }$ with probability $0.5$. This implies that the expected number of the $a_{ k, l }^{ i }$ and $b_{ k, l }^{ i }$ bits with value $1 (0)$ is $\frac { d } { 2 }$. Consequently, the expected number of tuples in $\mathbf { a }^{ i }$ (and in $\mathbf { b }_{ k }^{ i }$) in which the bit in the $k^{ th }$ position has value $1 (0)$ is $\frac { d } { 2 }$. Thus, irrespective of whether the verification outcome $c_{ k }^{ i }$ is $1$ or $0$, the expected number of tuples in $\mathbf { p }_{ k }^{ i }$ in which the bit in the $k^{ th }$ position has value $c_{ k }^{ i }$ is $\frac { d } { 2 }$, which proves property \eqref{eq: Expected Number of Tuples with Specific Value in Position k}. This also means that the expected number of the remaining tuples in $\mathbf { p }_{ k }^{ i }$, which are cryptic tuples according to Definition \ref{def: Proof Sequences}, is also $\frac { d } { 2 }$.
		\item	Measuring two pairs of qubits, both in the $\ket{ \Phi^{ + } }$ state, will result in both qubits of the first pair collapsing in state $\ket{ 0 }$ with probability $0.5$, or in state $\ket{ 1 }$ with probability $0.5$, and, independently, both qubits of the second pair collapsing in state $\ket{ 0 }$ with probability $0.5$, or in state $\ket{ 1 }$ with probability $0.5$. This means that the expected number of the $a_{ k, l }^{ i }$ and $a_{ k^{ \prime }, l }^{ i }$ bits with values ``$00$'', ``$01$'', ``$10$'', and ``$11$'' is $\frac { d } { 4 }$. Consequently, the expected number of tuples in $\mathbf { a }^{ i }$ in which the bits in positions $k$ and $k^{ \prime }$ contain any one of the aforementioned combinations is $\frac { d } { 4 }$. Thus, irrespective of whether the verification outcome $c_{ k }^{ i }$ is $1$ or $0$, the expected number of tuples in $\mathbf { p }_{ k }^{ i }$ in which the bits in positions $k$ and $k^{ \prime }$ are $c_{ k }^{ i } c_{ k }^{ i }$ or $c_{ k }^{ i } \overline { c_{ k }^{ i } }$ is $\frac { d } { 4 }$, which proves property \eqref{eq: Expected Number of Tuples with Specific Values in Positions k & k'}.
	\end{itemize}
	This completes the proof of this proposition.
\end{proof}

The properties outlined in Proposition \ref{prp: Is Alice's Proof Consistent?} are instrumental in the design of the verification tests that are used as subroutines for the QNVA. These tests, which are performed by every news aggregator Bob$_{ k }^{ i }$, $1 \leq k \leq n_{ i }$, in order to assess whether or not the coordinator Alice$_{ i }$ and the other news aggregators are honest, are based on the verification outcome $c_{ k }^{ i }$ and the proof sequence $\mathbf { p }_{ k }^{ i }$ that Bob$_{ k }^{ i }$ has received from Alice$_{ i }$.

As we have emphasized, our algorithm can handle all three scenarios mentioned above. For the first scenario (\textbf{S}$_{ 1 }$), the verification test \textsc{IsAlice'sProofConsistent} contained in Figure \ref{fig: The IsAlice'sProofConsistent Algorithm} can decide whether or not $\mathbf { p }_{ k }^{ i }$ is consistent with $c_{ k }^{ i }$ by checking if it satisfies Proposition \ref{prp: Is Alice's Proof Consistent?}. It relies on the auxiliary test \textsc{IsProofBalanced} shown below. It is essential to point out that in a real implementation of these tests one must take into account the possible imperfections of the quantum channel, and the probabilistic outcome of the measurements. That means that the stringent equality requirement of the expected values as expressed in the propositions should be relaxed and one should instead check for approximate equality $\approx$ or approximate inequality $\napprox$. In the presentation of the pseudocode, we adopt the following conventions.

\begin{itemize}
	\item	$i$, $1 \leq i \leq m$, is the index of Alice$_{ i }$
	\item	$k$, $1 \leq k \leq n_{ i }$, is the index of Bob$_{ k }^{ i }$
	\item	$c_{ k }^{ i }$ is the verification outcome that Alice$_{ i }$ sends to Bob$_{ k }^{ i }$
	\item	$\mathbf { p }_{ k }^{ i }$ is the proof sequence that Alice$_{ i }$ sends to Bob$_{ k }^{ i }$
	\item	$\mathbf { b }_{ k }^{ i }$ is the classical bit sequence of Bob$_{ k }^{ i }$
\end{itemize}

\begin{tcolorbox}
	[
		grow to left by = 0.00 cm,
		grow to right by = 0.00 cm,
		colback = WordVeryLightTeal!50,			
		enhanced jigsaw,						
		sharp corners,
		toprule = 0.10 pt,
		bottomrule = 0.10 pt,
		leftrule = 0.1 pt,
		rightrule = 0.1 pt,
		sharp corners,
		center title,
		fonttitle = \bfseries
	]
	\begin{figure}[H]
		\centering
		\begin{minipage}[b]{1.00 \textwidth}
			\begin{algorithm}[H]
				\SetAlgorithmName{Auxiliary Test}{ }{ }
				\renewcommand{\thealgocf}{}                                                         
				\caption{ \textsc{IsProofBalanced}$( i, k, c_{ k }^{ i }, \mathbf { p }_{ k }^{ i } )$ }
				\label{alg: The IsProofBalanced Algorithm}
				\For {$r ( \neq k ) = 1$ \KwTo $n_{ i }$}
				{
					$N_{ 1 }
					=
					\lvert
					\
					P_{ c_{ k }^{ i }, c_{ k }^{ i } }
					( \mathbf { p }_{ k }^{ i }, k, r )
					\
					\rvert
					$
					\\
					$N_{ 2 }
					=
					\lvert
					\
					P_{ c_{ k }^{ i }, \overline { c_{ k }^{ i } } }
					( \mathbf { p }_{ k }^{ i }, k, r )
					\
					\rvert
					$
					\\
					\If { $( N_{ 1 } \neq \frac { d } { 4 } \ \mathbf{ OR } \ N_{ 2 } \neq \frac { d } { 4 } )$ }
					{
						\Return FALSE
					}
				}
				\Return TRUE
			\end{algorithm}
			\caption{This auxiliary algorithm is invoked by Bob$_{ k }^{ i }$ during the main verification tests to ascertain whether property \eqref{eq: Expected Number of Tuples with Specific Values in Positions k & k'} of Proposition \ref{prp: Is Alice's Proof Consistent?} holds.}
			\label{fig: The IsProofBalanced Algorithm}
		\end{minipage}
	\end{figure}
\end{tcolorbox}

\begin{tcolorbox}
	[
		grow to left by = 0.00 cm,
		grow to right by = 0.00 cm,
		colback = WordVeryLightTeal!50,			
		enhanced jigsaw,						
		sharp corners,
		toprule = 0.10 pt,
		bottomrule = 0.10 pt,
		leftrule = 0.1 pt,
		rightrule = 0.1 pt,
		sharp corners,
		center title,
		fonttitle = \bfseries
	]
	\begin{figure}[H]
		\centering
		\begin{minipage}[b]{1.00 \textwidth}
			\begin{algorithm}[H]
				\SetAlgorithmName{Verification Test}{ }{ }
				\setcounter{algocf}{0}
				\caption{ \textsc{IsAlice'sProofConsistent}$( i, k, c_{ k }^{ i }, \mathbf { p }_{ k }^{ i }, \mathbf { b }_{ k }^{ i } )$ }
				\label{alg: The IsAlice'sProofConsistent Algorithm}
				$N
				=
				\lvert
				\
				P_{ c_{ k }^{ i } }
				( \mathbf { p }_{ k }^{ i }, k )
				\
				\rvert
				$
				\\
				\If { $N \neq \frac { d } { 2 }$ }
				{
					\Return FALSE
				}
				\ForEach { $l \in P_{ c_{ k }^{ i } } ( \mathbf { p }_{ k }^{ i }, k )$ }
				{
					\If { $( p_{ k, l }^{ i } \neq b_{ k, l }^{ i } )$ }
					{
						\Return FALSE
					}
				}
				\Return \textsc{IsProofBalanced}$( i, k, c_{ k }^{ i }, \mathbf { p }_{ k }^{ i } )$
			\end{algorithm}
			\caption{Bob$_{ k }^{ i }$ uses the above algorithm to check if the proof sequence $\mathbf { p }_{ k }^{ i }$ is consistent with the verification outcome $c_{ k }^{ i }$.}
			\label{fig: The IsAlice'sProofConsistent Algorithm}
		\end{minipage}
	\end{figure}
\end{tcolorbox}

To cope with the second scenario (\textbf{S}$_{ 2 }$), the verification test \textsc{IsBob'sProofConsistent} contained in Figure \ref{fig: The IsBob'sProofConsistent Algorithm} can decide whether or not $\mathbf { p }_{ k }^{ i }$ is consistent with $c_{ k }^{ i }$ by virtue of the results proved in the next proposition.

\begin{proposition} \label{prp: Is Bob's Proof Consistent?} \
	Suppose that Bob$_{ k }^{ i }$, $1 \leq k \leq n_{ i }$, has received from Alice$_{ i }$, $1 \leq i \leq m$, the verification outcome $c_{ k }^{ i } = 1$ $( c_{ k }^{ i } = 0 )$, and the consistent proof sequence $\mathds{ 1 }_{ k }^{ i }$ $( \vmathbb{ 0 }_{ k }^{ i } )$.

	Let us further assume that Bob$_{ k }^{ i }$ has also received from Bob$_{ k^{ \prime } }^{ i }$ ($k^{ \prime } \neq k$) the opposite verification outcome $c_{ k^{ \prime } }^{ i } = 0$ $( c_{ k^{ \prime } }^{ i } = 1 )$ and the sequence $\vmathbb{ 0 }_{ k^{ \prime } }^{ i }$ $( \mathds{ 1 }_{ k^{ \prime } }^{ i } )$ as proof. If $\vmathbb{ 0 }_{ k^{ \prime } }^{ i }$ $( \mathds{ 1 }_{ k^{ \prime } }^{ i } )$ is consistent with $0$ $( 1 )$, then, in addition to the properties listed in Proposition \ref{prp: Is Alice's Proof Consistent?}, it must also satisfy the following property.
	\begin{align}
		\left\{
		\
		\begin{matrix*}[l]
			P_{ 1, 0 }
			( \mathds{ 1 }_{ k }^{ i }, k, k^{ \prime } )
			=
			P_{ 1, 0 }
			( \vmathbb{ 0 }_{ k^{ \prime } }^{ i }, k, k^{ \prime } )
			\\
			\\
			P_{ 0, 1 }
			( \vmathbb{ 0 }_{ k }^{ i }, k, k^{ \prime } )
			=
			P_{ 0, 1 }
			( \mathds{ 1 }_{ k^{ \prime } }^{ i }, k, k^{ \prime } )
		\end{matrix*}
		\
		\right\}
		\label{eq: The Explicit Equality of Tuples with Opposite Values in Positions k & k'}
	\end{align}
\end{proposition}

\begin{proof}
	Without loss of generality we consider the situation that has evolved as follows.
	\begin{itemize}
		\item	Initially, Bob$_{ k }^{ i }$ received from Alice$_{ i }$ the verification outcome $c_{ k }^{ i } = 1$ and the consistent proof sequence $\mathds{ 1 }_{ k }^{ i }$, and
		\item	subsequently, Bob$_{ k }^{ i }$ received from Bob$_{ k^{ \prime } }^{ i }$ the opposite verification outcome $c_{ k^{ \prime } }^{ i } = 0$ and the sequence $\vmathbb{ 0 }_{ k^{ \prime } }^{ i }$ as proof.
	\end{itemize}
	We shall prove that if $\vmathbb{ 0 }_{ k^{ \prime } }^{ i }$ is consistent with the outcome $0$, then, in addition to the properties listed in Proposition \ref{prp: Is Alice's Proof Consistent?}, it must also satisfy the properties outlined above.

	The proof is an immediate consequence of the manner proof sequences are constructed. If we recall Definition \ref{def: Proof Sequences}, we see that the proof sequence $\mathds{ 1 }_{ k }^{ i }$ $( \vmathbb{ 0 }_{ k }^{ i } )$, which is consistent with the verification outcome $c_{ k }^{ i } = 1$ $( c_{ k }^{ i } = 0 )$, contains all the $n_{ i }$-tuples of Alice$_{ i }$'s bit sequence $\mathbf { a }^{ i }$ in which the bit in the $k^{ th }$ position has value $1$ $( 0 )$, including all those in which the bit in position $k^{ \prime }$ has the value $1$, and all those in which the bit in position $k^{ \prime }$ has the value $0$.
	\begin{align}
		\mathbf { a }^{ i }
		&=
		\underbrace { \colorbox {WordAquaLighter40} { $a_{ n_{ i }, d }^{ i } \dots a_{ k, d }^{ i } \dots a_{ 1, d }^{ i }$ } }_{ \text{ tuple } d }
		\cdots
		\underbrace { \colorbox {WordAquaLighter60} { $a_{ n_{ i }, 2 }^{ i } \dots a_{ k, 2 }^{ i } \dots a_{ 1, 2 }^{ i }$ } }_{ \text{ tuple } 2 }
		\underbrace { \colorbox {WordAquaLighter80} { $a_{ n_{ i }, 1 }^{ i } \dots a_{ k, 1 }^{ i } \dots a_{ 1, 1 }^{ i }$ } }_{ \text{ tuple } 1 }
		\ , \text{ and }
		\label{eq: Alice's Classical Bit Se quence}
		\\
		\mathbf { b }_{ k }^{ i }
		&=
		\underbrace { \colorbox {WordLightGreen} { $b_{ n_{ i }, d }^{ i } \dots b_{ k, d }^{ i } \dots b_{ 1, d }^{ i }$ } }_{ \text{ tuple } d }
		\cdots
		\underbrace { \colorbox {WordLightGreen!50} { $b_{ n_{ i }, 2 }^{ i } \dots b_{ k, 2 }^{ i } \dots b_{ 1, 2 }^{ i }$ } }_{ \text{ tuple } 2 } \
		\underbrace { \colorbox {WordLightGreen!25} { $b_{ n_{ i }, 1 }^{ i } \dots b_{ k, 1 }^{ i } \dots b_{ 1, 1 }^{ i }$ } }_{ \text{ tuple } 1 }
		\label{eq: Bob's Classical Bit Se quence}
		\ ,
	\end{align}
	Symmetrically, if the proof sequence $\vmathbb{ 0 }_{ k^{ \prime } }^{ i }$  $( \mathds{ 1 }_{ k^{ \prime } }^{ i } )$ is consistent with the opposite verification outcome $c_{ k^{ \prime } }^{ i } = 0$ $( c_{ k^{ \prime } }^{ i } = 1 )$, then it must contain all the $n_{ i }$-tuples of $\mathbf { a }^{ i }$ in which the bit in position $k^{ \prime }$ has value $0$ $( 1 )$, including all those in which the bit in position $k$ has the value $0$, and all those in which the bit in position $k$ has the opposite value $1$.

	Therefore, if both proof sequences $\mathds{ 1 }_{ k }^{ i }$ and $\vmathbb{ 0 }_{ k^{ \prime } }^{ i }$ ($\vmathbb{ 0 }_{ k }^{ i }$ and $\mathds{ 1 }_{ k^{ \prime } }^{ i }$) are consistent with the verification checks $c_{ k }^{ i } = 1$ and $c_{ k^{ \prime } }^{ i } = 0$ ($c_{ k }^{ i } = 0$ and $c_{ k^{ \prime } }^{ i } = 1$), respectively, then they must contain all the $n_{ i }$-tuples of $\mathbf { a }^{ i }$ in which the bit in the $k^{ th }$ position has the value $1$ ($0$) and the bit in position $k^{ \prime }$ has the opposite value $0$ ($1$). Formally, we can express this fact as
	\begin{align}
		\left\{
		\
		\begin{matrix*}[l]
			P_{ 1, 0 }
			( \mathds{ 1 }_{ k }^{ i }, k, k^{ \prime } )
			=
			P_{ 1, 0 }
			( \vmathbb{ 0 }_{ k^{ \prime } }^{ i }, k, k^{ \prime } )
			\\
			\\
			P_{ 0, 1 }
			( \vmathbb{ 0 }_{ k }^{ i }, k, k^{ \prime } )
			=
			P_{ 0, 1 }
			( \mathds{ 1 }_{ k^{ \prime } }^{ i }, k, k^{ \prime } )
		\end{matrix*}
		\
		\right\}
		\ ,
		\label{eq: Explicit Consistent Proofs Contain the Same Tuples with Opposite Values in Positions k & k'}
	\end{align}
	which concludes this proof.
	The opposite case is completely symmetrical and will be omitted.
\end{proof}

The previous proposition can be cast into its most general form as the following Corollary.

\begin{corollary} \label{crl: The Fundamental Property of Two Consistent Proof Sequences} \
	Let us assume that Bob$_{ k }^{ i }$, $1 \leq k \leq n_{ i }$, has received
	\begin{itemize}
		\item[$\Diamond$]	from Alice$_{ i }$, $1 \leq i \leq m$, the verification outcome $c_{ k }^{ i }$ and the sequence $\mathbf { p }_{ k }^{ i }$ as proof, and
		\item[$\Diamond$]	from Bob$_{ k^{ \prime } }^{ i }$ ($k^{ \prime } \neq k$) the opposite verification outcome $c_{ k^{ \prime } }^{ i } = \overline { c_{ k }^{ i } }$ and the sequence $\mathbf { p }_{ k^{ \prime } }^{ i }$ as proof.
	\end{itemize}
	Then, if both $\mathbf { p }_{ k^{ \prime } }^{ i }$ and $\mathbf { p }_{ k^{ \prime } }^{ i }$ are consistent with $c_{ k }^{ i }$ and $\overline { c_{ k }^{ i } }$, respectively, they must also satisfy the following property.
	\begin{align}
		P_{ c_{ k }^{ i }, \overline { c_{ k }^{ i } } }
		( \mathbf { p }_{ k }^{ i }, k, k^{ \prime } )
		=
		P_{ c_{ k }^{ i }, \overline { c_{ k }^{ i } } }
		( \mathbf { p }_{ k^{ \prime } }^{ i }, k, k^{ \prime } )
		\label{eq: The Equality of Tuples with Opposite Values in Positions k & k'}
	\end{align}
\end{corollary}

\begin{proof}
	If we recall Definition \ref{def: Proof Sequences} again, we see that if the proof sequence $\mathbf { p }_{ k }^{ i }$ is consistent with $c_{ k }^{ i }$, then it contains all the $n_{ i }$-tuples of Alice$_{ i }$'s bit sequence $\mathbf { a }^{ i }$ in which the bit in the $k^{ th }$ position has value $c_{ k }^{ i }$, including all those in which the bit in position $k^{ \prime }$ has the same value $c_{ k }^{ i }$, and all those in which the bit in position $k^{ \prime }$ has the opposite value $\overline { c_{ k }^{ i } }$. Symmetrically, if the proof sequence $\mathbf { p }_{ k^{ \prime } }^{ i }$ is consistent with $\overline { c_{ k }^{ i } }$, then it contains all the $n_{ i }$-tuples of $\mathbf { a }^{ i }$ in which the bit in position $k^{ \prime }$ has value $\overline { c_{ k }^{ i } }$, including all those in which the bit in position $k$ has the same value $\overline { c_{ k }^{ i } }$, and all those in which the bit in position $k$ has the opposite value $c_{ k }^{ i }$. Therefore, if they are consistent, both proof sequences $\mathbf { p }_{ k }^{ i }$ and $\mathbf { p }_{ k^{ \prime } }^{ i }$ contain all the $n_{ i }$-tuples of $\mathbf { a }^{ i }$ in which the bit in the $k^{ th }$ position has value $c_{ k }^{ i }$ and the bit in position $k^{ \prime }$ has the opposite value $\overline { c_{ k }^{ i } }$. This is simply written as
	\begin{align}
		P_{ c_{ k }^{ i }, \overline { c_{ k }^{ i } } }
		( \mathbf { p }_{ k }^{ i }, k, k^{ \prime } )
		=
		P_{ c_{ k }^{ i }, \overline { c_{ k }^{ i } } }
		( \mathbf { p }_{ k^{ \prime } }^{ i }, k, k^{ \prime } )
		\ ,
		\label{eq: Consistent Proofs Conatian the Same Tuples with Opposite Values in Positions k & k'}
	\end{align}
	which completes the proof of this corollary.
\end{proof}

To sum it up, the property expressed by relation \eqref{eq: The Equality of Tuples with Opposite Values in Positions k & k'} asserts that if two proof sequences $\mathbf { p }_{ k }^{ i }$ and $\mathbf { p }_{ k^{ \prime } }^{ i }$ that correspond to opposite outcomes $c_{ k }^{ i }$ and $\overline { c_{ k }^{ i } }$ are both consistent, then they must contain precisely the same tuples of $\mathbf { a }^{ i }$ in which the bit in the $k^{ th }$ position is $c_{ k }^{ i }$ and the bit in position $k^{ \prime }$ is $\overline { c_{ k }^{ i } }$. This property can be employed by Bob$_{ k }^{ i }$ to detect if Bob$_{ k^{ \prime } }^{ i }$ deliberately spreads misinformation, as formalized by the next Theorem \ref{thr: Does Bob Deliberately Spread Disinformation?}.

\begin{theorem} [Malicious news aggregator detection] \label{thr: Does Bob Deliberately Spread Disinformation?} \
	Suppose that Bob$_{ k }^{ i }$, $1 \leq k \leq n_{ i }$, has received from Alice$_{ i }$, $1 \leq i \leq m$, the verification outcome $c_{ k }^{ i }$, and the consistent proof sequence $\mathbf { p }_{ k }^{ i }$.

	Any attempt by another news aggregator Bob$_{ k^{ \prime } }^{ i }$ ($k^{ \prime } \neq k$) to falsely claim that he received $\overline { c_{ k }^{ i } }$ from Alice$_{ i }$, despite the fact that in reality he received $c_{ k }^{ i }$, and forge a proof sequence $\mathbf { p }_{ k^{ \prime } }^{ i }$ consistent with $\overline { c_{ k }^{ i } }$ will be detected by Bob$_{ k }^{ i }$.
\end{theorem}

\begin{proof}
	The present situation concerns how a malicious Bob$_{ k^{ \prime } }^{ i }$ may try to deceive Bob$_{ k }^{ i }$ ($k^{ \prime } \neq k$). Bob$_{ k^{ \prime } }^{ i }$ has received from Alice$_{ i }$ the verification outcome $c_{ k }^{ i }$ together with a proof sequence $\mathbf { p }_{ k^{ \prime } }^{ i }$ consistent with $c_{ k }^{ i }$. Nevertheless, Bob$_{ k^{ \prime } }^{ i }$ intends to falsely claim that he has received $\overline { c_{ k }^{ i } }$. The question is: can Bob$_{ k^{ \prime } }^{ i }$ construct a convincing proof sequence $\mathbf { p }_{ k^{ \prime } }^{ i }$ consistent with $\overline { c_{ k }^{ i } }$. We proceed to show that this is probabilistically impossible.

	Having received and validated $\mathbf { p }_{ k }^{ i }$, Bob$_{ k }^{ i }$ knows the set $P_{ c_{ k }^{ i }, \overline { c_{ k }^{ i } } } ( \mathbf { p }_{ k }^{ i }, k, k^{ \prime } )$ of the positions of all the $n_{ i }$-tuples of $\mathbf { a }^{ i }$ that contain $c_{ k }^{ i }$ and $\overline { c_{ k }^{ i } }$ in positions $k$ and $k^{ \prime }$ respectively.

	On the other hand, Bob$_{ k^{ \prime } }^{ i }$ has received the proof sequence $\mathbf { p }_{ k^{ \prime } }^{ i }$ that is also consistent with $c_{ k }^{ i }$. Accordingly, Bob$_{ k^{ \prime } }^{ i }$ knows the following two facts.
	\begin{enumerate}  [ left = 0.85 cm, labelsep = 0.50 cm, start = 1 ]
		\renewcommand\labelenumi{(\textbf{Fact}$_{ \theenumi }$)}
		\item	The indices of the $n_{ i }$-tuples of $\mathbf { a }^{ i }$ that contain $c_{ k }^{ i }$ in position $k^{ \prime }$, which includes those that also contain $c_{ k }^{ i }$ in position $k$, and those that also contain $\overline { c_{ k }^{ i } }$ in position $k$, i.e., the set
		\begin{align}
			P_{ c_{ k }^{ i } }
			( \mathbf { p }_{ k^{ \prime } }^{ i }, k^{ \prime } )
			=
			P_{ c_{ k }^{ i }, c_{ k }^{ i } }
			( \mathbf { p }_{ k^{ \prime } }^{ i }, k^{ \prime }, k )
			\cup
			P_{ c_{ k }^{ i }, \overline { c_{ k }^{ i } } }
			( \mathbf { p }_{ k^{ \prime } }^{ i }, k^{ \prime }, k )
			\ .
		\end{align}
		\item	The indices of the cryptic tuples $\mathbf { s }_{ \ast }$ of $\mathbf { a }^{ i }$, i.e., the set $P_{ \ast } ( \mathbf { p }_{ k^{ \prime } }^{ i } )$.
	\end{enumerate}
	By knowing the indices of the cryptic tuples, Bob$_{ k^{ \prime } }^{ i }$ is able to infer with certainty, i.e., probability $1$, that these indices correspond to $n_{ i }$-tuples of $\mathbf { a }^{ i }$ that contain $\overline { c_{ k }^{ i } }$ in position $k^{ \prime }$. In his effort to forge a proof sequence consistent with $\overline { c_{ k }^{ i } }$, Bob$_{ k^{ \prime } }^{ i }$ will correctly place all the tuples of $\mathbf { a }^{ i }$ that contain $\overline { c_{ k }^{ i } }$ in position $k^{ \prime }$. According to Proposition \ref{prp: Is Alice's Proof Consistent?}, their expected number is $\frac { d } { 2 }$, so in reality they would be $\approx \frac { d } { 2 }$. So, Bob$_{ k^{ \prime } }^{ i }$ will avoid trivial mistakes, such as
	\begin{itemize}
		\item	including a tuple where the bit in the $k^{ \prime }$ position has the wrong value, or
		\item	using fewer than expected tuples with $\overline { c_{ k }^{ i } }$ in position $k^{ \prime }$.
	\end{itemize}
	Bob$_{ k^{ \prime } }^{ i }$'s real weakness stems from the fact that the tuples he must include in his forged proof sequence may contain either $c_{ k }^{ i }$ with probability $0.5$, or $\overline { c_{ k }^{ i } }$ with equal probability $0.5$ in the $k^{ th }$ position because Bob$_{ k^{ \prime } }^{ i }$ doesn't know with certainty, even for a single tuple, if it contains $c_{ k }^{ i }$ or $\overline { c_{ k }^{ i } }$ in position $k$. Therefore, when forging a proof sequence consistent with $\overline { c_{ k }^{ i } }$, Bob$_{ k^{ \prime } }^{ i }$ has to guess for every tuple whether to place $c_{ k }^{ i }$ or $\overline { c_{ k }^{ i } }$ in position $k$. Thus, he is prone to make two types of mistakes.
	\begin{enumerate} [ left = 0.50 cm, labelsep = 0.50 cm, start = 1 ]
		\renewcommand\labelenumi{(\textbf{M}$_\theenumi$)}
		\item	Place $c_{ k }^{ i }$ in the $k^{ th }$ position of a wrong $n_{ i }$-tuple not contained in $P_{ c_{ k }^{ i }, \overline { c_{ k }^{ i } } } ( \mathbf { p }_{ k }^{ i }, k, k^{ \prime } )$.
		\item	Place $\overline { c_{ k }^{ i } }$ in the $k^{ th }$ position of a wrong $n_{ i }$-tuple that does appear in $P_{ c_{ k }^{ i }, \overline { c_{ k }^{ i } } } ( \mathbf { p }_{ k }^{ i }, k, k^{ \prime } )$.
	\end{enumerate}

	In other words, the question now becomes: how probable is for Bob$_{ k^{ \prime } }^{ i }$ to construct a proof sequence $\mathbf { p }_{ k^{ \prime } }^{ i }$ so that the set $P_{ c_{ k }^{ i }, \overline { c_{ k }^{ i } } } ( \mathbf { p }_{ k^{ \prime } }^{ i }, k, k^{ \prime } )$ is equal to the set $P_{ c_{ k }^{ i }, \overline { c_{ k }^{ i } } } ( \mathbf { p }_{ k }^{ i }, k, k^{ \prime } )$?

	The probability that Bob$_{ k^{ \prime } }^{ i }$ succeeds in doing so equals the probability of picking the one correct configuration out of many. The total number of configurations is equal to the number of ways to place $\approx \frac { d } { 4 }$ identical objects into $\approx \frac { d } { 2 }$ distinguishable boxes. Hence, the probability that Bob$_{ k^{ \prime } }^{ i }$ places \emph{all} the $\approx \frac { d } { 4 }$ values $c_{ k }^{ i }$ correctly in the $\approx \frac { d } { 2 }$ cryptic $n_{ i }$-tuples is
	\begin{align}
		P
		\left(
		\text{Bob$_{ k^{ \prime } }^{ i }$ places all $c_{ k }^{ i }$ correctly}
		\right)
		\approx
		\frac { 1 }
		{ \binom { \ d / 2 \ } { \ d / 4 \ } }
		\ ,
		\label{eq: Probability Bob k' Deceives Bob k}
	\end{align}
	which is practically zero for appropriately chosen values of $d$. Thus, the end result will violate property \eqref{eq: The Equality of Tuples with Opposite Values in Positions k & k'} of Corollary \ref{crl: The Fundamental Property of Two Consistent Proof Sequences}.

	Ergo, when Bob$_{ k }^{ i }$ checks the consistency of the proof sequence sent by Bob$_{ k^{ \prime } }^{ i }$ he will easily detect inconsistencies and infer that Bob$_{ k^{ \prime } }^{ i }$ deliberately spreads disinformation.
\end{proof}

So, to cope with the second scenario (\textbf{S}$_{ 2 }$), one may rely on the verification test \textsc{IsBob'sProofConsistent} shown in Figure \ref{fig: The IsBob'sProofConsistent Algorithm}, which can decide whether or not $\mathbf { p }_{ k^{ \prime } }^{ i }$ is consistent with $\overline { c_{ k }^{ i } }$, by checking if it satisfies property \eqref{eq: The Equality of Tuples with Opposite Values in Positions k & k'} of Corollary \ref{crl: The Fundamental Property of Two Consistent Proof Sequences}.
We again note that in a real implementation of the next test we must take into account the possible imperfections of the quantum channel, and the probabilistic outcome of the measurements, which implies that the strict inequality requirement should be relaxed and we should test for approximate inequality $\napprox$. In the pseudocode, we use the following conventions.

\begin{itemize}
	\item	$i$, $1 \leq i \leq m$, is the index of Alice$_{ i }$
	\item	$k$, $1 \leq k \leq n_{ i }$, is the index of Bob$_{ k }^{ i }$
	\item	$k^{ \prime }$, $1 \leq k^{ \prime } \neq k \leq n_{ i }$, is the index of Bob$_{ k^{ \prime } }^{ i }$
	\item	$c_{ k }^{ i }$ is the verification outcome that Alice$_{ i }$ has send to Bob$_{ k }^{ i }$
	\item	$\mathbf { p }_{ k }^{ i }$ is the proof sequence that Alice$_{ i }$ has send to Bob$_{ k }^{ i }$
	\item	$\overline { c_{ k }^{ i } }$ is the verification outcome that Bob$_{ k^{ \prime } }^{ i }$ claims he received from Alice$_{ i }$
	\item	$\mathbf { p }_{ k^{ \prime } }^{ i }$ is the proof sequence that Bob$_{ k^{ \prime } }^{ i }$ claims he received from Alice$_{ i }$
\end{itemize}

\begin{tcolorbox}
	[
		grow to left by = 0.00 cm,
		grow to right by = 0.00 cm,
		colback = WordVeryLightTeal!50,			
		enhanced jigsaw,						
		sharp corners,
		toprule = 0.10 pt,
		bottomrule = 0.10 pt,
		leftrule = 0.1 pt,
		rightrule = 0.1 pt,
		sharp corners,
		center title,
		fonttitle = \bfseries
	]
	\begin{figure}[H]
		\centering
		\begin{minipage}[b]{1.00 \textwidth}
			\begin{algorithm}[H]
				\SetAlgorithmName{Verification Test}{ }{ }
				\caption{\textsc{IsBob'sProofConsistent}$( i, k, k^{ \prime }, c_{ k }^{ i }, \mathbf { p }_{ k }^{ i }, \overline { c_{ k }^{ i } }, \mathbf { p }_{ k^{ \prime } }^{ i } )$ }
				\label{alg: The IsBob'sProofConsistent Algorithm}
				$N
				=
				\lvert
				\
				P_{ \overline { c_{ k }^{ i } } } ( \mathbf { p }_{ k^{ \prime } }^{ i }, k^{ \prime } )
				\
				\rvert
				$
				\\
				\If { $N \neq \frac { d } { 2 }$ }
				{
					\Return FALSE
				}
				\If
				{
					$P_{ c_{ k }^{ i }, \overline { c_{ k }^{ i } } } ( \mathbf { p }_{ k }^{ i }, k, k^{ \prime } )
					\neq
					P_{ c_{ k }^{ i }, \overline { c_{ k }^{ i } } } ( \mathbf { p }_{ k^{ \prime } }^{ i }, k, k^{ \prime } )$
				}
				{
					\Return FALSE
				}
				\Return \textsc{IsProofBalanced}$( i, k^{ \prime }, \overline { c_{ k }^{ i } }, \mathbf { p }_{ k^{ \prime } }^{ i } )$
			\end{algorithm}
			\caption{Bob$_{ k }^{ i }$ uses the above algorithm to check if $\mathbf { p }_{ k^{ \prime } }^{ i }$ is consistent with $\overline { c_{ k }^{ i } }$ that Bob$_{ k^{ \prime } }^{ i }$ claims to have received from Alice$_{ i }$.}
			\label{fig: The IsBob'sProofConsistent Algorithm}
		\end{minipage}
	\end{figure}
\end{tcolorbox}

By combining both verification checks, it is possible to detect an insidious Alice$_{ i }$ who deliberately spreads disinformation and confusion by sending opposite verification outcomes $c_{ k }^{ i }$ and $\overline { c_{ k }^{ i } }$ to Bob$_{ k }^{ i }$ and Bob$_{ k^{ \prime } }^{ i }$ ($k^{ \prime } \neq k$), using the correct proof sequences $\mathbf { p }_{ k }^{ i }$ and $\mathbf { p }_{ k^{ \prime } }^{ i }$ in each case. This is analyzed in the following Theorem \ref{thr: Does Alice Deliberately Spread Disinformation?}.

\begin{theorem} [Malicious news verifier detection] \label{thr: Does Alice Deliberately Spread Disinformation?} \
	Suppose that Bob$_{ k }^{ i }$, $1 \leq k \leq n_{ i }$, has received from Alice$_{ i }$, $1 \leq i \leq m$, the verification outcome $c_{ k }^{ i }$ and the consistent proof sequence $\mathbf { p }_{ k }^{ i }$.

	Bob$_{ k }^{ i }$ infers that Alice$_{ i }$ is a malicious actor that deliberately spreads disinformation, if he also receives a proof sequence $\mathbf { p }_{ k^{ \prime } }^{ i }$ consistent with the opposite outcome $\overline { c_{ k }^{ i } }$ from another news aggregator Bob$_{ k^{ \prime } }^{ i }$ ($k^{ \prime } \neq k$).
\end{theorem}

\begin{proof}
	The present situation examines how a news aggregator can uncover an insidious news verifier Alice$_{ i }$ who deliberately spreads disinformation and confusion by sending the verification outcome $c_{ k }^{ i }$ to Bob$_{ k }^{ i }$ and, at the same time, sending the opposite verification outcome $\overline { c_{ k }^{ i } }$ to Bob$_{ k^{ \prime } }^{ i }$ ($k^{ \prime } \neq k$), using consistent proof sequences $\mathbf { p }_{ k }^{ i }$ and $\mathbf { p }_{ k^{ \prime } }^{ i }$ in each case.

	According to Theorem \ref{thr: Does Bob Deliberately Spread Disinformation?}, the probability that another news aggregator Bob$_{ k^{ \prime } }^{ i }$ ($k^{ \prime } \neq k$) will manage to construct on his own a proof sequence $\mathbf { p }_{ k^{ \prime } }^{ i }$ consistent is negligible. Hence, if the verification test \textsc{IsBob'sProofConsistent} shown in Figure \ref{fig: The IsBob'sProofConsistent Algorithm} returns TRUE, the logical conclusion is that Alice$_{ i }$ herself must have sent the consistent proof sequence $\mathbf { p }_{ k^{ \prime } }^{ i }$ to Bob$_{ k^{ \prime } }^{ i }$.

	Thus, Alice$_{ i }$ deliberately sends contradictory verification outcomes to create confusion and spread disinformation.
\end{proof}

At this point, considering all the previous analysis, we present the proposed quantum news verification algorithm (QNVA) below. For every piece of news that must be checked, the QNVA is employed by each news aggregator Bob$_{ k }^{ i }$, $1 \leq k \leq n_{ i }$, independently and in parallel with every other news aggregator. In the presentation, we use the following notation.

\begin{itemize}
	\item	$i$, $1 \leq i \leq m$, is the index of Alice$_{ i }$
	\item	$k$, $1 \leq k \leq n_{ i }$, is the index of Bob$_{ k }^{ i }$
	\item	QNVA$( k )$ is the instance of QVNA executed by Bob$_{ k }^{ i }$
	\item	$M_{ A }$ and $M_{ V }$ are the list of malicious news aggregators and news verifiers, respectively, as surmised by Bob$_{ k }^{ i }$. The purpose of the reputation lists is to identity insidious agents and ignore any further communication originating from them.
	\item	$k^{ \prime }$, $1 \leq k^{ \prime } \neq k \leq n_{ i }$, is the index of Bob$_{ k^{ \prime } }^{ i }$
	\item	$c_{ k }^{ i }$ is the verification outcome that Alice$_{ i }$ has send to Bob$_{ k }^{ i }$
	\item	$\mathbf { p }_{ k }^{ i }$ is the proof sequence that Alice$_{ i }$ has send to Bob$_{ k }^{ i }$
	\item	$c_{ k^{ \prime } }^{ i }$ is the verification outcome that Bob$_{ k^{ \prime } }^{ i }$ claims he received from Alice$_{ i }$
	\item	$\mathbf { p }_{ k^{ \prime } }^{ i }$ is the proof sequence that Bob$_{ k^{ \prime } }^{ i }$ claims he received from Alice$_{ i }$
\end{itemize}

\begin{tcolorbox}
	[
		enhanced,
		breakable,
		grow to left by = 0.00 cm,
		grow to right by = 0.00 cm,
		colback = WordVeryLightTeal!50,			
		enhanced jigsaw,						
		sharp corners,
		toprule = 0.01 pt,
		bottomrule = 0.01 pt,
		leftrule = 0.1 pt,
		rightrule = 0.1 pt,
		sharp corners,
		center title,
		fonttitle = \bfseries
	]
	\begin{algorithm}[H]
		\SetAlgorithmName{Algorithm}{ }{ }
		\setcounter{algocf}{0}
		\caption { \textsc { QNVA$( k )$ } }
		\label{alg: The Quantum News Verification Algorithm}
	\end{algorithm}
	\begin{enumerate}  [ left = 0.90 cm, labelsep = 0.50 cm, start = 0 ]
		\renewcommand\labelenumi{(\textbf{Step}$_\theenumi$)}
		\item	\textbf{Initialize} $\triangleright$ $M_{ A } = M_{ V } = \emptyset$
		\item	\textbf{Receive} $\triangleright$ Bob$_{ k }^{ i }$ receives Alice$_{ i }$'s verification outcome $c_{ k }^{ i }$ and proof $\mathbf { p }_{ k }^{ i }$.
		\item	\textbf{Test} $\triangleright$ Bob$_{ k }^{ i }$ calls the verification test \textsc{IsAlice'sProofConsistent} (Figure \ref{fig: The IsAlice'sProofConsistent Algorithm}) to check whether $\mathbf { p }_{ k }^{ i }$ is consistent with $c_{ k }^{ i }$.
		\begin{itemize}
			\item[$\star$]	If the test returns TRUE, then Bob$_{ k }^{ i }$ accepts Alice$_{ i }$'s assessment.
			\item[$\star$]	If the test returns FALSE, then Bob$_{ k }^{ i }$ rejects the news in question as fake, adds Alice$_{ i }$ to his $M_{ V }$ list, and terminates the algorithm.
		\end{itemize}
		\item	\textbf{Send} $\triangleright$ Upon the successful completion of the previous verification check, Bob$_{ k }^{ i }$ sends every other Bob$_{ k^{ \prime } }^{ i }$ ($1 \leq k^{ \prime } \neq k \leq n_{ i }$) not contained in his $M_{ A }$ list, the verification outcome $c_{ k }^{ i }$ and the accompanying proof $\mathbf { p }_{ k }^{ i }$ received from Alice$_{ i }$.
		\item	\textbf{Receive} $\triangleright$ Bob$_{ k }^{ i }$ receives from every other Bob$_{ k^{ \prime } }^{ i }$ ($1 \leq k^{ \prime } \neq k \leq n_{ i }$) not contained in his $M_{ A }$ list, the verification outcome $c_{ k^{ \prime } }^{ i }$ and proof $\mathbf { p }_{ k^{ \prime } }^{ i }$ Bob$_{ k^{ \prime } }^{ i }$ claims he received from Alice$_{ i }$.
		\item	\textbf{Compare} $\triangleright$ Bob$_{ k }^{ i }$ compares his $c_{ k }^{ i }$ to all other $c_{ k^{ \prime } }^{ i }$.
		\begin{itemize}
			\item[$\star$]	If all $c_{ k^{ \prime } }^{ i }$ coincide with $c_{ k }^{ i }$, then Bob$_{ k }^{ i }$ sticks to his preliminary decision, and terminates the algorithm.
			\item[$\star$]	If there is \emph{at least one} $c_{ k^{ \prime } }^{ i }$ such that $c_{ k^{ \prime } }^{ i } = \overline { c_{ k }^{ i } }$, Bob$_{ k }^{ i }$ calls the verification test \textsc{IsBob'sProofConsistent} (Figure \ref{fig: The IsBob'sProofConsistent Algorithm}) to check whether $\mathbf { p }_{ k^{ \prime } }^{ i }$ is consistent with $\overline { c_{ k }^{ i } }$.
			\begin{itemize}
				\item[$\square$]	If the test returns FALSE, then Bob$_{ k }^{ i }$ adds Bob$_{ k^{ \prime } }^{ i }$ to his $M_{ A }$ list, and repeats the same procedure for the next opposite verification outcome, if any.
				\item[$\square$]	If the test returns TRUE, then Bob$_{ k }^{ i }$ rejects the news in question as fake, adds Alice$_{ i }$ to his $M_{ V }$ list, and terminates the algorithm.
			\end{itemize}
		\end{itemize}
	\end{enumerate}
\end{tcolorbox}

In real life, the existence of opposite conflicting verification outcomes increases the odds of confusion and unchecked spread of misinformation. The quantum news verification algorithm, by taking advantage of the phenomenon of entanglement and its unique ramifications, can eliminate the risks in certain critical situations, as those outline in the preceding scenarios (\textbf{S}$_{ 1 }$) -- (\textbf{S}$_{ 3 }$).

\section{Discussion and conclusions} \label{sec: Discussion and Conclusions}

In the era of social media, the proliferation of fake news has emerged as a pressing issue. Particularly in economically developed countries, users tend to encounter more false information than accurate content. The impact of fake news on major social media platforms extends beyond the digital realm, influencing people’s opinions and actions in the real world. Researchers have been driven to seek practical solutions to address this undesirable situation.

This research paper introduces a fresh perspective on the critical topic of news verification. Departing from the conventional Quantum Machine Learning approach, our approach explores an alternative quantum avenue. Drawing inspiration from successful quantum protocols that achieve distributed and detectable Byzantine Agreement in massively distributed environments, we propose the entanglement-based quantum algorithm QNVA.

The QNVA offers several advantages:

\begin{itemize}
	\item	\textbf{Generality:} It can handle any number of news aggregators and verifiers.
	\item	\textbf{Efficiency:} The algorithm completes in a constant number of steps, regardless of the participant count.
	\item	\textbf{Simplicity:} It relies solely on EPR (specifically $\ket{ \Phi^{ + } }$) pairs. EPR pairs are the easiest maximally entangled states to produce, unlike more complex states such as $\ket{ GHZ_{ n } }$, which do not scale well as the number of players increases.
\end{itemize}

The aforementioned attributes underscore its scalability and practical applicability. To reinforce this assertion, we examine in Table \ref{tbl: Numerical Characteristics of the QNVA} how the chosen value of the accuracy degree $d$ influences the likelihood of a malicious aggregator successfully fabricating a consistent proof sequence. Notably, the accuracy degree $d$ remains independent of the number of participants, further enhancing the algorithm's scalability. Naturally, selecting an appropriate value for $d$ is crucial to ensure the negligible probability of a malicious actor successfully forging a consistent proof sequence. The rationale behind $d$ not scaling with the number of aggregators and verifiers lies in the protocol's consistent utilization of EPR pairs, signifying bipartite entanglement. As per protocol guidelines, each consistency check involves a comparison between two bit vectors. Consequently, irrespective of the participant count, each comparison entails only two bit strings. Furthermore, even in the most general scenario, this comparison involves just two bits, denoted as $i$ and $j$, in each tuple. Thus, probabilistically, the situation remains consistent. In essence, the probability of a malicious aggregator deceiving an honest aggregator hinges on the likelihood of selecting the correct configuration from many possibilities. The total number of configurations equals the ways to distribute approximately $\frac{d}{4}$ identical objects (either $0$ or $1$) into approximately $\frac{d}{2}$ distinguishable boxes (representing the uncertain tuples). The probability of a cheater correctly placing \emph{all} the approximately $\frac{d}{4}$ bits within the approximately $\frac{d}{2}$ cryptic tuples is

\begin{align}
	P( \text{ malicious aggregator cheats } )
	\approx
	\frac { 1 }
	{ \binom { \ d / 2 \ } { \ d / 4 \ } }
	\ , \label{eq: Malicious Aggregator Cheats}
\end{align}
which tends to zero as $d$ increases.

\begin{tcolorbox}
	[
	grow to left by = 0.00 cm,
	grow to right by = 0.00 cm,
	colback = WordVeryLightTeal!25,			%
	enhanced jigsaw,						
	sharp corners,
	boxrule = 0.1 pt,
	toprule = 0.1 pt,
	bottomrule = 0.1 pt
	]
	\begin{table}[H]
		\renewcommand{\arraystretch}{1.60}
		\caption{This table shows how the chosen value of the degree of accuracy $d$ affects the probability that a malicious aggregator succeeds in forging a consistent proof sequence.}
		\label{tbl: Numerical Characteristics of the QNVA}
		\centering
		\begin{tabular}
			{
				>{\centering\arraybackslash} m{1.00 cm} !{\vrule width 0.5 pt}
				>{\centering\arraybackslash} m{1.50 cm} !{\vrule width 0.5 pt}
				>{\centering\arraybackslash} m{1.50 cm} !{\vrule width 0.5 pt}
				>{\centering\arraybackslash} m{7.00 cm}
			}
			\Xhline{4\arrayrulewidth}
			\multicolumn{4}{c}
			{ \cellcolor[HTML]{000000} { \color[HTML]{FFFFFF} \textbf{How the degree of accuracy $d$ affects the probability $P$} } }
			\\
			\Xhline{\arrayrulewidth}
			{ \cellcolor[HTML]{000000} { \color[HTML]{FFFFFF} $d$ } }
			&
			$d / 2 $
			&
			$d / 4 $
			&
			$P( \text{ malicious aggregator cheats } )$
			\\
			\Xhline{3\arrayrulewidth}
			{ \cellcolor[HTML]{000000} { \color[HTML]{FFFFFF} $4$ } }
			&
			\xinteval { 4 // 2 }
			&
			\xinteval { 4 // 4 }
			&
			\xintdeffloatfunc	Combinations ( m ) := ( m / 2 )! // ( ( m / 4 )! ( m / 4 )! );
			\xintdeffloatfunc	Probability ( m ) := 1 / Combinations ( m );
			\xintfloateval{ Probability ( 4 ) }
			\\
			\Xhline{\arrayrulewidth}
			{ \cellcolor[HTML]{000000} { \color[HTML]{FFFFFF} $8$ } }
			&
			\xinteval { 8 // 2 }
			&
			\xinteval { 8 // 4 }
			&
			\xintDigits*:=5;
			\xintdeffloatfunc	Combinations ( m ) := ( m / 2 )! // ( ( m / 4 )! ( m / 4 )! );
			\xintdeffloatfunc	Probability ( m ) := 1 / Combinations ( m );
			\xintfloateval{ Probability ( 8 ) }
			\\
			\Xhline{\arrayrulewidth}
			{ \cellcolor[HTML]{000000} { \color[HTML]{FFFFFF} $16$ } }
			&
			\xinteval { 16 // 2 }
			&
			\xinteval { 16 // 4 }
			&
			\xintDigits*:=5;
			\xintdeffloatfunc	Combinations ( m ) := ( m / 2 )! // ( ( m / 4 )! ( m / 4 )! );
			\xintdeffloatfunc	Probability ( m ) := 1 / Combinations ( m );
			\xintfloateval{ Probability ( 16 ) }
			\\
			\Xhline{\arrayrulewidth}
			{ \cellcolor[HTML]{000000} { \color[HTML]{FFFFFF} $32$ } }
			&
			\xinteval { 32 // 2 }
			&
			\xinteval { 32 // 3 }
			&
			\xintDigits*:=5;
			\xintdeffloatfunc	Combinations ( m ) := ( m / 2 )! // ( ( m / 4 )! ( m / 4 )! );
			\xintdeffloatfunc	Probability ( m ) := 1 / Combinations ( m );
			\xintfloateval{ Probability ( 32 ) }
			\\
			\Xhline{\arrayrulewidth}
			{ \cellcolor[HTML]{000000} { \color[HTML]{FFFFFF} $64$ } }
			&
			\xinteval { 64 // 3 }
			&
			\xinteval { 64 // 4 }
			&
			\xintDigits*:=5;
			\xintdeffloatfunc	Combinations ( m ) := ( m / 2 )! // ( ( m / 4 )! ( m / 4 )! );
			\xintdeffloatfunc	Probability ( m ) := 1 / Combinations ( m );
			\xintfloateval{ Probability ( 64 ) }
			\\
			\Xhline{4\arrayrulewidth}
		\end{tabular}
		\renewcommand{\arraystretch}{1.0}
	\end{table}
\end{tcolorbox}
\bibliographystyle{ieeetr}
\bibliography{QNVfNA}

\end{document}